\documentclass[11pt]{article}

\usepackage{etex, authblk}
\usepackage{algorithmic, algorithm}
\usepackage{amsmath,amsfonts,amssymb,tikz}
\usepackage[margin=1in]{geometry}
\usepackage{amsthm} % includes the proof environment
\usepackage{csquotes}
\usepackage[all]{xy}
\usepackage[bottom]{footmisc}
\usepackage{enumitem}
\usepackage{prettyref} % Use this package to refer to theorems, conjectures and etc
\usepackage{hyperref}
\usepackage{bbm, bm}
\usepackage{footnote}

\newtheorem{theorem}{Theorem}[section]

\newtheorem{lemma}[theorem]{Lemma}
\newtheorem{corollary}[theorem]{Corollary}

\newtheorem{claim}[theorem]{Claim}
\newtheorem{proposition}[theorem]{Proposition}
\newtheorem{definition}[theorem]{Definition}
\newtheorem{remark}[theorem]{Remark}

\newtheorem{fact}[theorem]{Fact}

\makeatletter
\newtheorem*{rep@theorem}{\rep@title}
\newcommand{\newreptheorem}[2]{%
\newenvironment{rep#1}[1]{%
\def\rep@title{#2 \ref{##1}}%
\begin{rep@theorem}}%
{\end{rep@theorem}}}
\makeatother

\newreptheorem{theorem}{Theorem}
\newreptheorem{lemma}{Lemma}
\newreptheorem{definition}{Definition}

{\hspace*{\fill}$\Box$\par}
\newenvironment{proofof}[1]{\smallskip\noindent{\bf Proof of #1.}}%
{\hspace*{\fill}$\Box$\par}

\newcommand{\pref}{\prettyref}
\newrefformat{lem}{Lemma \ref{#1}}
\newrefformat{cl}{Claim \ref{#1}}
\newrefformat{prop}{Proposition \ref{#1}}
\newrefformat{prob}{Problem \ref{#1}}
\newrefformat{thm}{Theorem \ref{#1}}
\newrefformat{cor}{Corollary \ref{#1}}
\newrefformat{cha}{Chapter \ref{#1}}
\newrefformat{sec}{Section \ref{#1}}
\newrefformat{tab}{Table \ref{#1}}
\newrefformat{fig}{Figure \ref{#1}}
\newrefformat{conj}{Conjecture \ref{#1}}
\newrefformat{prot}{Protocol \ref{#1}}
\newrefformat{fact}{Fact \ref{#1}}
\newrefformat{def}{Definition \ref{#1}}

\newcommand{\E}{{\mathbb{E}}}
\newcommand{\eps}{\varepsilon}

\newcommand{\cB}{\mathcal{B}}

\newcommand{\cD}{\mathcal{D}}

\newcommand{\cI}{\mathcal{I}}
\newcommand{\cP}{\mathcal{P}}
\newcommand{\cG}{\mathcal{G}}
\newcommand{\cM}{\mathcal{M}}

\newcommand{\cX}{\mathcal{X}}

\newcommand{\cQ}{\mathcal{Q}}
\newcommand{\cU}{\mathcal{U}}
\newcommand{\cT}{\mathcal{T}}
\newcommand{\cF}{\mathcal{F}}
\newcommand{\cJ}{\mathcal{J}}

\newcommand{\Patrascu}{P\u{a}tra\c{s}cu}

\newcommand{\ip}[1]{\left\langle #1 \right\rangle}

\newcommand{\Field}{\mathbb{F}}

\newcommand{\abs}[1]{\left| #1 \right|}
\newcommand{\norm}[1]{\| #1 \|}

\newcommand{\KL}[2]{\mathsf{D}_{KL} ( {#1} || {#2} )}

\usepackage{xspace}

\newcommand{\poly}{{\operatorname{poly}\xspace}}

\newcommand{\AND}{\mathsf{AND}}
\newcommand{\OR}{\mathsf{OR}}

\newcommand{\supp}{\mathsf{supp}}
\newcommand{\odisc}{\mathsf{odisc}}

\newcommand{\adv}{ \overrightarrow{\mathsf{adv}}}

\newcommand{\IP}{\mathsf{IP}}

\DeclareMathOperator*{\argmax}{arg\,max}

\floatstyle{boxed}
\newfloat{Protocol}{H}{pro}

\floatstyle{boxed}
\newfloat{Problem}{H}{pro}
\title{Unifying the Landscape of Super-Logarithmic \\
Dynamic Cell-Probe Lower Bounds}
\author{Young Kun Ko}
\affil[]{Department of Computer Science and Engineering, Pennsylvania State University}
\affil[]{Email: ykko@psu.edu}
\date{\today}

\begin{document}

\maketitle

\begin{abstract}
    We prove a general translation theorem for converting one-way communication lower bounds over a product distribution to dynamic cell-probe lower bounds. 
    
    Specifically, we consider a class of problems considered in \cite{patrascu_mihai_towards_2010} where:
    \begin{itemize}
        \item $S_1, \ldots, S_m \in \{0, 1\}^n$ are given and publicly known.
        \item $T \in \{0, 1\}^n$ is a sequence of updates, each taking $t_u$ time.
        \item For a given $Q \in [m]$, we must output $f(S_Q, T)$ in $t_q$ time.
    \end{itemize}
    Our main result shows that for a ``hard'' function $f$, for which it is difficult to obtain a non-trivial advantage over random guessing with one-way communication under some product distribution over $S_Q$ and $T$ (for example, a uniform distribution), then the above explicit dynamic cell-probe problem must have $\max \{ t_u, t_q \} \geq \tilde{\Omega}(\log^{3/2}(n))$ if $m = \Omega(n^{0.99})$. This result extends and unifies the super-logarithmic dynamic data structure lower bounds from \cite{larsen_crossing_2020} and \cite{larsen_super-logarithmic_2025} into a more general framework.

    From a technical perspective, our approach merges the cell-sampling and chronogram techniques developed in \cite{larsen_crossing_2020} and \cite{larsen_super-logarithmic_2025} with the new static data structure lower bound methods from \cite{ko_adaptive_2020} and \cite{ko_lower_2025}, thereby merging all known state-of-the-art cell-probe lower-bound techniques into one. 

    As a direct consequence of our method, we establish a super-logarithmic lower bound against the Multiphase Problem \cite{patrascu_mihai_towards_2010} for the case where the data structure outputs the Inner Product (mod 2) of $S_Q$ and $T$. We suspect further applications of this general method towards showing super-logarithmic dynamic cell-probe lower bounds. We list some example applications of our general method, including a novel technique for a one-way communication lower bound against small-advantage protocols for a product distribution using average min-entropy, which could be of independent interest.
\end{abstract}

\newpage

\section{Introduction}

We consider the following general class of dynamic data structure problems, initiated by the so-called Multiphase Program \cite{patrascu_mihai_towards_2010}.
\begin{itemize}
    \item $S_1, \ldots, S_m \in \{ 0, 1 \}^{n}$ are given and publicly known.\footnote{In the original Multiphase Problem, these are given as pre-processing inputs and are pre-processed into a data structure.}
    \item $T \in \{ 0, 1 \}^{n}$ is given as a sequence of updates, using $t_u$ time per update.
    \item For any given $i \in [m]$, output $f ( S_i, T )$ using $t_q$ time.
\end{itemize}
We analyze this dynamic data structure problem in the cell-probe model \cite{yao_should_1981}, the most powerful model of computation for data structures. An input is pre-processed into a data structure of $s$ cells, each with a word of size $w$-bits. In this model, we only charge for memory accesses (probes) of cells, while all computation on top of the probed cells is free. Because of its unimaginable strength, a lower bound in this model applies to any reasonable data structure.

The Multiphase Conjecture \cite{patrascu_mihai_towards_2010, thorup_mihai_2013} states that if $f$ is two-party Disjointness between $S_i$ and $T$ with $m = \poly (n)$, then $\max \{ t_u, t_q \} \geq n^{\delta}$ for some constant $\delta > 0$. The Multiphase Conjecture then implies polynomial lower bounds for Graph Reachability (for directed graphs), Dynamic Shortest Path (for undirected graphs), and other interesting dynamic problems. 

In this work, we consider a generalized version of the Multiphase Program, considering $f$ to be any function from a general class of functions. For instance, we consider $f$ to be a ``lifted" version of some function $\psi$. Suppose we divide $n$ into blocks of size $k$ each, denoted by $S_i [ j ], T [ j ]$ where $j \in [n / k]$. Let $g$ be some inner gadget, $g: \{0,1\}^k \times \{0,1\}^k \to \{ 0 , 1 \}.$ Some standard inner gadgets $g$ include indexing functions (where the first $k$ bits have only one 1) or functions that output the inner product of the two $k$-bit strings. Then consider $f$'s of the form
$$f ( S_i, T ) := \psi ( g ( S_i [ 1 ] , T [ 1 ] ) , \ldots, g ( S_i [ n/k ] , T [ n/k ] ).$$
For example, this captures both Disjointness and Inner Product (mod 2), by setting $k=1$ and $g$ as the bitwise-AND function. If we take $\psi$ as $\bigvee$, $f$ is then Disjointness, while if we take $\psi$ as $\bigoplus$, $f$ is then Inner Product (mod 2). This is the class of functions studied extensively in so-called lifting theorems (see \cite{rao_anup_communication_2018, bun_approximate_2022} and references therein).

\subsection{Our Result}

Our main result is to show a ``lifting" type result for the dynamic data structure problem. If $f$ is ``hard" against an $n / \poly \log(n)$ length message (i.e., a lower bound against one-way communication) under a product distribution, then the Multiphase Problem for such a function $f$ has $\max \{ t_u, t_q \} \geq \tilde{\Omega} ( \log^{3/2} (n) )$. Thus we lift $f$ which is hard against one-way communication (a more tractable task) to hard against dynamic data structures, albeit only up to super-logarithmic hardness. As a corollary, we obtain the state-of-the-art lower bound against the Multiphase Conjecture \cite{patrascu_mihai_towards_2010, clifford_new_2015, brody_adapt_2015, ko_adaptive_2020, dvorak_lower_2020, ko_lower_2025} when the function in question is Inner Product (mod 2) instead of Disjointness, improving upon the previously known $\Omega ( \log (n) )$ bound \cite{clifford_new_2015, brody_adapt_2015}.

Formally, we consider the following set of explicit dynamic data structure problems, which we denote as the generalized Multiphase Problem:
\begin{Problem}
    \begin{itemize}
    \item $S_1, \ldots, S_m \in \{ \{0 , 1 \}^k \}^{n/k}$ are given, where each $S_Q$ is divided into $n/k$ blocks $S_Q [1], \ldots, S_Q [n/k]$, each of size $k$.
    \item $T = ( T [1], \ldots, T [n/k] ) \in \{ \{0 , 1 \}^k \}^{n/k} $ is given as a sequence of updates, using $t_u$ time per update (in the cell-probe model with word size $w$).
    \item For any given $Q \in \cQ = [m]$, output $$f(S_Q, T) = \psi \left( g( S_Q [1], T [1] ), \ldots , g( S_Q [n/k], T [n/k] ) \right)$$ in $t_q$ time.
    \end{itemize}
    \caption{Generalized Multiphase Problem \label{prob:gmp} }
\end{Problem}
Note that the functions $f$ we consider are precisely the set of functions used in the Pattern Matrix Method \cite{sherstov_pattern_2011}, and Lifting Theorems (\cite{rao_anup_communication_2018, bun_approximate_2022} and references therein). Our main result is showing super-logarithmic lower bounds (as in \cite{larsen_crossing_2020, larsen_super-logarithmic_2025}) for a general class of ``hard" $f$'s which includes inner product (mod 2). 

Before fully stating our main result, we would like to formally describe which $f$'s are ``hard" for our proof. As we also use the {\bf chronogram method}, we will divide the $n/k$ blocks of $T$ into $\ell$ epochs, $\{ T_i \}_{i=1}^\ell$, where $T_i$ consists of $n_i$ blocks and $|T_i| = k \cdot n_i \approx k \cdot \gamma^i.$ $\gamma$ is a parameter to be defined within the proof, and $\sum n_i = n/k$. If we fix all epochs other than $i$, this results in our function $f$ becoming
\begin{equation*}
    f_i ( S_Q, T_i )  := f ( S_Q, T )
\end{equation*}
With fixed $S_Q^{-i}$ and $T_{-i}$, we consider the real-valued matrix $\Psi_i |_{S_Q^{-i}, T_{-i}} \in [ - 1, + 1]^{ 2^{|S_Q^i|} \times  2^{|T_i|} }$ (assuming without loss of generality that $\Pr_{S_Q^i, T_i} [ f ( S_Q, T ) = +1] \geq \Pr_{S_Q^i, T_i} [ f ( S_Q, T ) = -1]$)
\begin{equation*}
    \Psi_i (S_Q^i, T_i ) := \begin{cases}
        \frac{1}{\Pr_{T_i} [ f ( S_Q, T ) = +1]} - 1 & \mbox{if } f ( S_Q, T ) = +1 \\
       - 1  & \mbox{otherwise}  
    \end{cases} 
\end{equation*}
Observe that $\Psi_i$ has been normalized so that for every setting of $S_Q^i$, $\E_{T_i} \left[ \Psi_i (S_Q^i, T_i ) \right] = 0$.

We need two technical components regarding $\Psi_i$. First, we cannot have $ \min_{S_Q^i, T_i} \abs{ \Psi_i (S_Q^i, T_i ) } $ be too small, which corresponds to having a lower bound of $\beta$ (we will set the bound on $\beta$ later in the section) on
$$ \min \left\{ \frac{\Pr_{T_i} [ f ( S_Q, T ) = -1]}{\Pr_{T_i} [ f ( S_Q, T ) = +1]} , \frac{\Pr_{T_i} [ f ( S_Q, T ) = +1]}{\Pr_{T_i} [ f ( S_Q, T ) = -1]} \right\} \geq \beta $$
Otherwise, we cannot use the simulation theorem from \cite{larsen_super-logarithmic_2025}. This is equivalent to having the function somewhat balanced, i.e., 
$$ \abs{ \E_{T_i} \left[ f_i ( S^i_Q, T_i ) \right] } \leq \frac{1 - \beta}{1 + \beta} $$

As the next component, we would like $\Psi_i$ to be resilient against a short message. If a short message dependent only on $T_i$ is sent, we would like the corresponding distribution on $\Psi_i$ to have small discrepancy. This results in the following definition.
\begin{definition}[One-way Discrepancy]
Fix $S_Q^{-i}, T_{-i}$. We define one-way discrepancy of $\Psi_i$ under the distribution $\mu$ over $T_i$ as 
\begin{equation*}
    \odisc_\mu ( \Psi_i ):= \E_{S^i_Q} \abs{ \sum_{T_i}  \Psi_i ( S^i_Q, T_i ) \cdot \mu (T_i) } 
\end{equation*}    
with the parameter $C > 0$ as follows
\begin{equation*}
    \odisc_C ( \Psi_i ):= \max_{ M(T_i): |M| < C }  \E_{M} \left[ \odisc_{T_i | M} ( \Psi_i ) \right]
\end{equation*}   
\end{definition}

An intuitive interpretation of one-way discrepancy is to use the following communication game. Suppose Alice is given $S_Q^i$ and Bob is given $T_i$, distributed independently (i.e. a product distribution). Bob generates a message $M$ of length at most $C$ (note that the message is independent of $S_Q^i$). Alice then generates a guess $\{ \pm 1 \}$ that maximizes the advantage in \cite{larsen_super-logarithmic_2025} for the balanced version of $f$. \footnote{If the function is a balanced function, $\Psi_i$ is simply a $\pm 1$ matrix (or the usual communication matrix for $f_i$).} Since Bob's message is independent of Alice's input, we can simply take the expectation over $S_Q^i$ to measure the expected advantage. The intuition of a ``hard" function is $f$ which is resilient against a short message by Bob. 

The above model is not new. Such a communication model was originally studied in \cite{kremer_randomized_1995}, that is, the lower bound for one-way communication under a product distribution but in the small constant-error regime.\footnote{\cite{kremer_randomized_1995} made a connection between the VC-dimension $d$ of $f_i$ with fixed $S^i_Q$ in the small constant-error regime. However, their bound does not apply in the low-advantage regime.} Here, we want to consider the low-advantage (high-error) regime.

The aforementioned two conditions can be summarized in the following definition of a hard $f$: 
\begin{definition} \label{def:hard}
    We say $f$ is ``hard" if for each $i \in [ \ell / 3 , \ell]$, with probability $\geq g_i \geq \Omega(1)$ over $S_Q^{-i}$ and $T_{-i}$, $f_i$ and its corresponding $\Psi_i$ satisfy both
    \begin{itemize}
        \item (Balanced) $\Pr_{S^i_Q} \left[ \abs{ \E_{T_i} \left[ f_i ( S^i_Q, T_i )  \right] }    \geq \frac{1 - \beta}{1+ \beta} \right] \leq  n^{-2}$ with $\beta \geq 2^{- o ( \sqrt{ \log n })}$
        \item (Low Discrepancy) $\odisc_C \left( \Psi_i \right) \leq n^{-2} $ for $C \leq n_i / \poly \log (n)$
    \end{itemize}
\end{definition}
Intuitively, this corresponds to saying that $f$ is ``hard" if for any given epoch $i \in [ \ell / 3 , \ell]$, the corresponding $f_i$ is balanced and resistant against a short one-way message with good probability over the remaining coordinates $S_Q^{-i}$ and $T_{-i}$. 

Then for such a hard $f$, we prove the following super-logarithmic lower bound on $\max \{ t_u, t_q \}$.
\begin{theorem}[Informal] 
    For the explicit dynamic data structure \pref{prob:gmp} equipped with a hard $f$, the update time $t_u$ and query time $t_q$ with $m = \Omega(n^{0.99})$ must have
    \begin{equation*}
        t_q \geq \Omega \left( \frac{ \log^{3/2} n }{\log^2 (w t_u)} \right).
    \end{equation*}
\end{theorem}
Therefore, our work essentially reduces proving a super-logarithmic dynamic cell-probe lower bound to proving a lower bound against one-way communication over a product distribution (i.e., showing that the function $f$ is ``hard").

\subsubsection{Comparison with Previous Results}

Progress on dynamic cell-probe lower bounds has been slow but steady over the past three decades. The seminal work of Fredman and Saks~\cite{fredman_cell_1989} provided the first $\Omega ( \log n / \log \log n )$ dynamic cell-probe lower bound. It took another 15 years to remove the $\log \log n$ factor from the denominator. \Patrascu~and Demaine \cite{patrascu_lower_2004, patrascu_logarithmic_2006} gave an $\Omega ( \log n )$ bound for an explicit dynamic problem. Yet achieving a super-logarithmic lower bound (i.e., $\omega( \log n)$ time lower bound on update/query time) seemed elusive~\cite{thorup_mihai_2013}.

Larsen~\cite{larsen_cell_2012} gave an $\Omega ( ( \log n / \log \log n )^2 )$ bound for the 2-dimensional range sum problem, breaking the logarithmic barrier for dynamic data structure lower bounds. Larsen's result, however, has a caveat. The number of output bits per query is $\Theta ( \log n )$. Therefore, $\max \{ t_u, t_q \}$ {\em per output bit} is still $\tilde{\Omega} ( \log n )$. It took roughly an extra decade to finally provide a super-logarithmic lower bound for a Boolean dynamic data structure problem, established in \cite{larsen_crossing_2020}, which gave an $\tilde{\Omega} ( \log^{3/2} n )$ lower bound for the dynamic 2-dimensional range parity sum problem. \cite{larsen_super-logarithmic_2025} then further extended the technique to dynamic graph $s$-$t$ reachability (in a directed acyclic graph).

\paragraph{Multiphase Program}

On the other hand, the most promising avenue of attack in achieving a super-logarithmic lower bound (or even a polynomial bound) on $\max \{ t_u, t_q \}$ for the past decades has been the Multiphase Program \cite{patrascu_mihai_towards_2010, thorup_mihai_2013}, the holy grail in dynamic cell-probe lower bounds.

\Patrascu's original approach for the Multiphase Program was to provide a lower bound for the following communication game.
\begin{Problem}
    \begin{itemize}
    \item Alice is given all the inputs but $T \in \{0,1\}^n$. Bob is given all the inputs but $\vec{S}$ consisting of $S_1, \ldots, S_m \in \{0 ,1 \}^n$. Merlin is given all the inputs but $Q \in [m]$. 
    \item Merlin sends a message of length $t_u \cdot n \cdot w$ to Bob. Alice and Bob then proceed in a standard two-party communication protocol to output $f(S_Q,T)$ after communicating $t_q w$ bits. 
    \end{itemize}
    \caption{Multiphase Communication Game}
\end{Problem}
There are results directly attacking the above communication game (namely, \cite{chattopadhyay_little_2012, ko_adaptive_2020, dvorak_lower_2020, ko_lower_2025}). But we emphasize that attacking the communication game is a much harder problem than the underlying dynamic data structure problem~\cite{ko_adaptive_2020, ko_lower_2025}. For instance, the Multiphase Communication Game has connections to a very old circuit lower bound problem \cite{jukna_circuits_2010}. 

Furthermore, there is a known separation between the communication model and the dynamic data structure problem in question \cite{ko_lower_2025}. There exists $f$ (namely indexing) that is easy in the communication model, but hard under the dynamic cell-probe model. 

But even if we turn to a weaker dynamic cell-probe lower bound, only a logarithmic lower bound is known. When $f$ is Disjointness or Inner Product (mod 2), \cite{brody_adapt_2015, clifford_new_2015} showed $\max \{ t_u, t_q \} \geq \Omega \left( \log n \right)$.\footnote{From a technical perspective, this is under a weaker model, where the query algorithm has no knowledge of $\vec{S}$ and therefore must probe for it as well. The setting we consider is slightly stronger as the querier already knows $\vec{S}$.}

In summary, even though the Multiphase Problem is conjectured to be harder than problems considered in \cite{larsen_crossing_2020, larsen_super-logarithmic_2025}, no super-logarithmic lower bounds were known for the Multiphase Problem prior to our work, due to a technical challenge that we explain in \pref{sec:technical}. 

\subsection{Technical Contribution} \label{sec:technical}

\subsubsection{One way simulation theorem of \cite{larsen_crossing_2020}}

One technical ingredient of our work is the one-way simulation theorem from \cite{larsen_crossing_2020, larsen_super-logarithmic_2025} which pioneered ``a recipe" for super-logarithmic dynamic cell-probe lower bounds. In this section, we highlight the contributions in \cite{larsen_crossing_2020, larsen_super-logarithmic_2025} to provide context for our technical contribution.

Consider a fixed dynamic data structure problem $f : \cQ \times T \to \{ \pm 1 \}$, where $\cQ$ is the set of queries, and $T$ is the sequence of updates. Therefore if $T$ is given as the sequence of updates, when $Q \in \cQ$ is given as the query, the data structure must output $f ( Q, T )$.

The first building block for the one-way simulation theorem is the seminal chronogram technique of Fredman and Saks \cite{fredman_cell_1989}. A simple way to interpret chronogram techniques is as follows. As the dynamic data structure must be ready to handle queries at any point during the sequence of updates, in contrast to the static data structure, we can divide the update $T$ into epochs of length $n_1, \ldots, n_\ell$ such that $\sum_{i=1}^\ell n_i = n/k$, where each $n_i = \gamma n_{i-1}$ for some parameter $\gamma$. Therefore, they are geometrically decreasing. The main observations in the chronogram technique are as follows: (i) The average number of cells that are probed which were last updated in epoch $i$ is $t_{tot} / \ell$; (ii) The number of cells changed by all subsequent epochs is small, as the number of updates decreases geometrically over the epochs. 

Now we focus on a fixed epoch $i$ to give a lower bound on $t_{tot} / \ell$. 
Such a static data structure (for a fixed epoch $i$) has pre-initialized memory (updates performed in previous epochs) and a cache (updates performed in subsequent epochs). Observe that a data structure in epoch $i$, with $n_i$ updates, which would generate $t_u n_i$ many updated cells, has $\sum_{j > i } t_u n_j \leq o ( n_i ) $ cells in cache, and $\sum_{j < i } t_u n_j$ cells in pre-initialized memory. 

To prove a lower bound for such a static data structure, \cite{larsen_crossing_2020, larsen_super-logarithmic_2025} introduce a one-way communication game which can simulate a static data structure with pre-initialized memory and a cache, to then argue that no ``too-good-to-be-true" one-way protocol can exist.
\begin{Protocol}
    \begin{itemize}
        \item Bob is given all the updates $\{ T_i \}_{i=1}^\ell$, but not the desired query $Q \in \cQ$. 
        \item Alice is given all the updates, except the updates in epoch $i$ (i.e., $T_i$), and the desired query.
        \item Bob sends a one-way message $M$ of length $C$-bits to Alice. Then Alice announces $f ( Q, T )$. 
    \end{itemize}
    \caption{One-way Communication Simulation $G_f$}
\end{Protocol}
Alice, without Bob's message, can only generate the pre-initialized memory on her own. 
In the extreme case where $C \geq w ( t_u n_i + o ( n_i ) )$, Bob can then send over the entire cache (which is at most $o(n_i)$ cells) and the memory state generated in epoch $i$. Therefore, Alice will announce the correct $f ( Q, T )$ every time. In fact, just sending $T_i$ would allow Alice to recover $f ( Q, T )$ for all $Q \in \cQ$. The goal is then to analyze the correctness of the protocol when $C \ll n_i$.

As a measure of correctness of the one-way communication protocol, \cite{larsen_super-logarithmic_2025} introduces the following quantity called the advantage of the protocol over the product distribution on the query set and update set $\cD = \cD_{\cQ} \times \cD_{T}$. 
\begin{equation*}
    \adv ( G_{f}, \cD_{\cQ}, \cD_{T}, C ) := \max_{M} \E_{Q, M} \left[ \abs{ \E_{T \sim \cD_{T} | M} \left[ \overline{f} ( Q, T ) | M \right] }  \right]
\end{equation*}
where $\overline{f} : \cQ \times T \to [ -1, +1 ]$ is the normalized version of $f$ so that for every $Q \in \cQ$, $\E_{T \sim \cD_{T}} [\overline{f} ( Q, T ) ] = 0$. Note that the normalization is only necessary when $f$ is not a balanced function for every $Q$. If $f$ is completely balanced, i.e., the setting for \cite{larsen_crossing_2020}, no normalization is necessary, and the above quantity exactly captures the small advantage over a fair random coin toss.

The remaining high-level argument of both \cite{larsen_crossing_2020, larsen_super-logarithmic_2025} is as follows: (i) show that a ``too-good-to-be-true" data structure implies a ``too-good-to-be-true" one-way communication protocol with a short one-way message (small $C$) and good advantage (i.e., the simulation theorem); (ii) show that such a ``too-good-to-be-true" one-way communication protocol cannot exist. 

As the main contribution in \cite{larsen_crossing_2020, larsen_super-logarithmic_2025} is the one-way simulation theorem (i.e., step (i)), let us focus on obtaining (i). To obtain (i), we would like to generate a short one-way message $M$ from the static data structure. The key technique to reduce the length of the one-way message $M$ is the cell-sampling technique \cite{siegel_universal_2004, panigrahy_lower_2010}. A naive attempt is to sample each cell updated in epoch $i$ with probability $1/(t_u w)^{\Theta(1)}$ independently at random. Bob then sends the sampled cells to Alice as the single-shot message, along with the cache. If all cells required to be read have been sampled, then Alice can answer correctly. If not, she just randomly guesses the answer. Since there is a non-negligible probability of sampling all required cells, this would result in some small advantage. However, the technical challenge is that Alice cannot distinguish between the above two cases just from the sampled cells alone. Alice cannot distinguish whether a cell she received from Bob's message has been touched during epoch $i$ or not.

The main technical centerpiece to resolve the above issue is the so-called ``Peak-to-Average" lemma \cite{larsen_crossing_2020}, which states the existence of a small subset of cells such that knowing their contents gives a nontrivial advantage. Using the lemma (which can be interpreted as a clever choice of cells for cell-sampling) and the chronogram technique \cite{fredman_cell_1989}, any ``too-good-to-be-true" dynamic data structure yields a static data structure (with pre-initialized memory and a cache) that achieves a small advantage over random guessing. 

The key insights in \cite{larsen_cell_2012, larsen_crossing_2020, larsen_super-logarithmic_2025} are to carefully subsample cells updated in epoch $i$ such that knowing the contents of such cells gives a non-trivial advantage over random guessing. The cache and sub-sampled cells constitute a short one-way message $M$ that achieves a non-trivial advantage over random guessing.

All these technical components formally culminate in the following one-way simulation theorem for a static data structure with pre-initialized memory and a cache.
\begin{theorem}[\cite{larsen_crossing_2020, larsen_super-logarithmic_2025}] \label{thm:ly_simulation}
    Let $f : \cQ \times T \to [-1, +1]$ be a data structure problem with weights $[-1,+1]$. $\cD_{\cQ} \times \cD_{T}$ be a (product) distribution over the queries and inputs such that $\E_{T \sim \cD_{T}} [ f( Q, T ) ]= 0$ for any $Q \in \cQ$, and $\abs{f(Q,T)} \geq \beta$ for some $\beta>0$. If there exists a data structure with pre-initialized memory and a cache for $f$ with at most $S$ updated cells, $S_{cac}$ cells in cache, expected query time $t_{tot}$, and expected query time into updated cells $t_q$, then for any $p \in (0, 1)$,
    \begin{align*}
        & \adv ( G_{\cP}, \cD_{\cQ}, \cD_{\cI}, ( 8 p S + S_{cac} ) w ) \\
        & \geq \exp \left( - O \left ( \log ( 1/p) ( t_q + \sqrt{ t_{tot} ( t_q \log (1/p) + \log (1 / \beta) )}  + \log (1 / \beta) \right)  \right) - \exp \left( - \Omega( p S ) \right)
    \end{align*}
\end{theorem}

\subsubsection{Our technical centerpiece} 

Our main point of departure from \cite{larsen_crossing_2020, larsen_super-logarithmic_2025} is (ii): showing that no such one-way communication protocol can exist.

The one-way communication game $G_f$ used in \cite{larsen_crossing_2020, larsen_super-logarithmic_2025} is based on an {\em explicit} static data structure lower bound on the Butterfly Graph \cite{patrascu_unifying_2011}. While this was comparatively straightforward for \cite{larsen_crossing_2020}, the primary technical challenge in \cite{larsen_super-logarithmic_2025} lies in proving a new static lower bound on the Butterfly Graph, as this requires establishing a strong {\em explicit} static data structure lower bound even for a small advantage over random guessing. This approach follows from a standard encoding-decoding (and therefore a counting) argument, inspired by the static data structure lower bound from \cite{patrascu_unifying_2011}.

A key distinction of our method is how it overcomes the limitations of the encoding-decoding (or counting) arguments used in prior work \cite{larsen_crossing_2020, larsen_super-logarithmic_2025} to establish one-way communication lower bounds. This approach is effective when the query distribution, $\cM$, is uniform over a small support, as a simple counting argument suffices (see, e.g., Section 5.2 of \cite{larsen_super-logarithmic_2025}). However, the encoding-decoding strategy fails when the query set is sub-sampled from a distribution with a much larger support. In this scenario, the one-way message can depend arbitrarily on the specific queries chosen, adding a layer of complexity that makes a simple counting argument intractable. This challenge was the primary bottleneck to applying the simulation lemma from \cite{larsen_crossing_2020} to the Multiphase Problem and necessitated our different approach.

More specifically, we need a lower bound for the following modified one-way communication game where both Alice and Bob have access to $\vec{S}$.
\begin{Protocol}
    \begin{itemize}
        \item Bob is given $\vec{S} = S_1, \ldots, S_m$ and all the updates $\{ T_j \}_{j=1}^\ell$, but not the desired query $Q \in \cQ = [m]$.
        \item Alice is given $\vec{S} = S_1, \ldots, S_m$, all updates except those in epoch $i$ (i.e., $T_i$), and the desired query $Q$.
        \item Bob sends a one-way message $M$ of length $C$ bits to Alice. Then Alice announces $f(S_Q, T)$.
    \end{itemize}
    \caption{Modified One-way Communication Game \label{prot:moc}}
\end{Protocol}
The main technical challenge for the encoding-decoding argument is the dependence on $\vec{S}$. Alice and Bob both know $\vec{S}$, and $M$ can be interpreted differently depending on the specific $\vec{S}$ that is chosen. Thus, we must account for all possible $\vec{S}$, which essentially invalidates the counting argument. The issue is compounded when proving the lower bound in the low-advantage regime, i.e., where the advantage is $n^{-\Omega(1)}$.

Our technical contribution is the use of information-theoretic tools from \cite{ko_adaptive_2020, ko_lower_2025} to essentially ``remove" the dependency on $\vec{S}$ when $S_i$ and $T$'s are generated under a product distribution. Our work reduces the problem to a lower bound for the following one-way communication game, which has been studied in various contexts and is far more tractable.
\begin{Protocol}
    \begin{itemize}
        \item Bob is given all the updates $\{ T_j \}_{j=1}^{\ell}$.
        \item Alice is given $\varsigma$, chosen from some distribution $\cD_{\cQ}$, and all updates except those in epoch $i$ (i.e., $T_i$).
        \item Bob sends a one-way message $M$ of length $C = n_i / \poly \log (n)$ bits to Alice. Then Alice announces $f(\varsigma, T)$.
    \end{itemize}
    \caption{Reduced One-way Communication Game \label{prot:reduced}}
\end{Protocol}
We remark that our definition of ``hard" (\pref{def:hard}) exactly captures the lower bound for \pref{prot:reduced}. Hard functions do not have a short protocol with good advantage for \pref{prot:reduced}.
In \pref{sec:multiphase_application} and \pref{sec:0xor}, we list some examples of these hard functions including the usual Inner Product mod 2. We remark that we need a lower bound of the form $C \geq n / \poly \log ( 1 / \eps )$ for an $\eps$-advantage (under a product distribution). We cannot expect such a lower bound for well-studied Indexing or Disjointness, as a simple upper bound with $n / \poly ( 1 / \eps )$ exists. We introduce a new method of lower bounding \pref{prot:reduced} against small advantage under a product distribution using average min-entropy.

\section{Preliminary}

\subsection{Information Theory}

In this section, we provide the necessary background on information theory and information complexity that are used 
in this paper. For further reference, we refer the reader to \cite{cover_elements_2006}.

\begin{definition}[Entropy]
	The entropy of a random variable $X$ is defined as 
	\begin{equation*}
	H(X) := \sum_{x} \Pr[X=x] \log \frac{1}{\Pr[X=x]}.
	\end{equation*}
	Similarly, the conditional entropy is defined as
	\begin{equation*}
	H(X|Y) := \E_{Y} \left[ \sum_{x} \Pr[X = x | Y = y] \log \frac{1}{\Pr[ X = x | Y = y]} \right].
	\end{equation*}
\end{definition}

\begin{fact}[Conditioning Decreases Entropy] \label{fact:conditioningentropy}
For any random variable $X$ and $Y$
\begin{equation*}
    H(X) \geq H(X|Y)
\end{equation*}
\end{fact}

\noindent With entropy defined, we can also quantify the correlation between two random variables, or how much information one random variable conveys about the other.

\begin{definition}[Mutual Information]
	Mutual information between $X$ and $Y$ (conditioned on $Z$) is defined as 
	\begin{equation*}
	I( X; Y | Z) := H(X|Z) - H(X|YZ).
	\end{equation*}
\end{definition}

\noindent Similarly, we can also define how much one distribution conveys information about the other distribution.

\begin{definition}[KL-Divergence]
	KL-Divergence between two distributions $\mu$ and $\nu$ is defined as
	\begin{equation*}
	\KL{\mu}{\nu} := \sum_{x} \mu(x) \log \frac{\mu(x)}{\nu(x)}.
	\end{equation*}
\end{definition}

\noindent To bound mutual information, it suffices to bound KL-divergence, due to the following fact.

\begin{fact}[KL-Divergence and Mutual Information] \label{fact:KL_Mutual}
	The following equality between mutual information and KL-Divergence holds 
	\begin{equation*}
	I(A;B|C) = \E_{B,C} \left[ \KL{ A|_{B=b, C=c} }{ A|_{C=c} } \right].
	\end{equation*}
\end{fact}

% \noindent Also the following relation is known between $\ell_1$ distance between two distributions and KL-Divergence. 

\begin{fact}[Pinsker's Inequality] \label{fact:pinsker}
For any two distributions $P$ and $Q$, 
\begin{equation*}
\norm{ P - Q }_{TV} = \frac{1}{2} \norm{ P - Q }_1 \leq \sqrt{ \frac{1}{2 \log e} D( P || Q) }	
\end{equation*}
\end{fact}

\noindent We also make use of the following facts on Mutual Information throughout the paper.

\begin{fact}[Chain Rule] \label{fact:chainrule}
For any random variable $A,B,C$ and $D$  
\begin{equation*}
    I(AD;B|C) = I(D;B|C) + I(A;B|CD).
\end{equation*}
\end{fact}

\begin{fact} \label{fact:chainrule1}
For any random variable $A,B,C$ and $D$, if $I(B;D|C) = 0$
\begin{equation*}
    I(A;B|C) \leq I(A;B|CD).
\end{equation*}
\end{fact}
\begin{proof} By the chain rule and non-negativity of mutual information, 
\begin{align*}
I(A;B|C) \leq I(AD;B|C) = I(B;D|C) + I(A;B|CD) = I(A;B|CD).
\end{align*}
\end{proof}

\begin{fact}\label{fact:chainrule2}
For any random variable $A,B,C$ and $D$, if $I(B;D|AC) = 0$
\begin{equation*}
    I(A;B|C) \geq I(A;B|CD).
\end{equation*}
\end{fact}
\begin{proof} By the chain rule and non-negativity of mutual information,  
\begin{align*}
I(A;B|CD) \leq I(AD;B|C) = I(A;B|C) + I(B;D|AC) = I(A;B|C).
\end{align*}
\end{proof}

\section{Proof of Main Theorem}

Recall that the class of explicit dynamic problems that we would like to give a lower bound for is the following.
\begin{Problem}
    \begin{itemize}
    \item $S_1, \ldots, S_m \in \{ \{0 , 1 \}^k \}^{n/k}$ are given, each $S_Q$ divided into $n/k$ blocks $S_Q [1], \ldots, S_Q [n/k]$, each of size $k$.
    \item $T = ( T [1], \ldots, T [n/k] ) \in \{ \{0 , 1 \}^k \}^{n/k} $ is given as a sequence of updates, using $t_u$ time per update.
    \item For any given $Q \in \cQ = [m]$, output $$f(S_Q, T) = \psi \left( g( S_Q [1], T[1] ), \ldots , g( S_Q [n/k], T [n/k] ) \right)$$ in $t_q$ time.
    \end{itemize}
    % \caption{\pref{prob:gmp}}
\end{Problem}
\noindent which we denote as the generalized Multiphase Problem. The original Multiphase Problem can be easily represented as $k=1$ with $g$ being the standard bit-wise $\AND$, $\psi$ being $\bigvee$, or $\bigoplus$ if $f$ is computing inner product over $\Field_2$.

We formally prove the following result for the generalized Multiphase Problem.
\begin{theorem} \label{thm:main_formal}
    Suppose $f$ is hard (as in \pref{def:hard}) and $m = \Omega( n^{0.99} )$. Then it must be the case that
    \begin{equation*}
        t_q \geq \Omega \left( \frac{ \log^{3/2} n }{\log^2 (w t_u)} \right).
    \end{equation*}
\end{theorem}

The main intuition behind phrasing our result in a generalized manner is that our result ``lifts" some hard function $\psi$ (in a restricted model) to a super-logarithmic data structure lower bound.

\paragraph{Epochs} To give a super-logarithmic lower bound, the first main technical component is the chronogram technique \cite{fredman_cell_1989}, where we divide the sequence of $n/k$ updates into epochs, each containing $n_i := \gamma^i$ blocks (of $k$ bits each), and $\sum_{i=1}^\ell n_i = n/k$. We denote the epochs of updates as $\{ T_i \}_{i=1}^\ell$, and we process the updates in reverse order.

Then we also consider the associated coordinates in $S_Q$. In particular, we denote $S_Q^i$ as the part of $S_Q$ associated with the $i$-th epoch $T_i$. This notation can be further extended to work with multiple indices in $[m]$ and $[\ell]$. $S_Q^{-i}$ denotes $S_Q$ except $S_Q^i$, analogously $T_{-i}$. We use $\vec{S}$ to denote the bundle over all $Q$'s. Similarly, if we use a subset of $[m]$ in the subscript, that denotes $S_Q$'s in the subscript.

\paragraph{Notations and Hard Distribution} 
We will use $\tau_{tot}$ as the random variable for the number of cells (total, across all epochs) probed in the query algorithm. $\tau_q$ is used as the random variable for the number of cells updated in epoch $i$ that are probed.

We assume each $S_Q [ i ]$ is distributed i.i.d. for all $i \in [n/k]$ and $Q \in [m]$, as well as all the coordinates in $T$. This implies that $S_Q [i]$'s and $T_i$'s are all independent. 

We will use the overline notation to denote the normalized version of $f$. That is, with fixed $S_Q^{-i}, T_{-i}$ assuming without loss of generality that $\Pr_{S_Q^i, T_i} [ f ( S_Q, T ) = +1] \geq \Pr_{S_Q^i, T_i} [ f ( S_Q, T ) = -1]$,
\begin{equation*}
    \overline{f}_i (S_Q^i, T_i ) := \begin{cases}
        \frac{1}{\Pr_{T_i} [ f ( S_Q, T ) = +1]} - 1 & \mbox{if } f ( S_Q, T ) = +1 \\
       - 1  & \mbox{otherwise}  
    \end{cases} 
\end{equation*}
and vice-versa if $\Pr_{S_Q^i, T_i} [ f ( S_Q, T ) = +1] \leq \Pr_{S_Q^i, T_i} [ f ( S_Q, T ) = -1]$.

\paragraph{High Level Proof Strategy} 
Before delving into the technical part, we would like to give an overview of our overall proof strategy.

\begin{enumerate}
    \item First, we want to select a ``good" epoch $i \in [\ell/3, 5\ell/6]$. We show that there exists some epoch $i$ of the updates, and an associated event $E_i$ and query set $\cQ'_i \subset \cQ$ which yield a ``good" data structure problem with pre-initialized memory and a cache.
    \item Such a ``good" data structure then yields a ``too-good-to-be-true" protocol for the following communication model from \cite{larsen_crossing_2020, larsen_super-logarithmic_2025}.
    \begin{figure}[!h]
            \centering
            \includegraphics[width=0.8\linewidth]{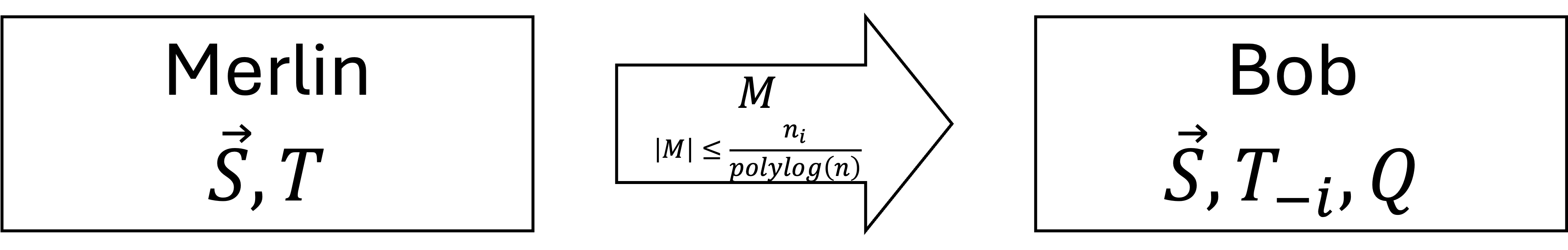}
            \caption{Communication Model}
    \end{figure}
    Bob sends a one-way message $M$ to Alice. Then Alice guesses the value of $f(S_Q, T) = f_i ( S_Q^i, T_i )$. The performance of the protocol is then measured by the advantage of the protocol.
    \item We show that a ``too-good-to-be-true" dynamic data structure implies a ``too-good-to-be-true" one-way simulation in \pref{sec:simulation}. Then in \pref{sec:lower_bound}, we show that hard $f$'s cannot have a ``too-good-to-be-true" one-way simulation using information-theoretic tools, thereby resulting in a contradiction.
\end{enumerate}

\subsection{One-way simulation Theorem} \label{sec:simulation}

In this section, we will prove the following simulation theorem (an analog of \pref{thm:ly_simulation} from \cite{larsen_crossing_2020, larsen_super-logarithmic_2025}) which converts a ``too-good-to-be-true" dynamic data structure for $f$ to a one-way communication protocol with a high advantage. Recall that our $f$ in question satisfies the following property. 

\begin{repdefinition}{def:hard}
    We say $f$ is ``hard" if for each $i \in [ \ell / 3 , \ell]$, with probability $\geq g_i \geq \Omega(1)$ over $S_Q^{-i}$ and $T_{-i}$, $f_i$ and its corresponding $\Psi_i$ satisfy both
    \begin{itemize}
        \item (Balanced) $\Pr_{S^i_Q} \left[ \abs{ \E_{T_i} \left[ f_i ( S^i_Q, T_i )  \right] }    \geq \frac{1 - \beta}{1+ \beta} \right] \leq  n^{-2}$ with $\log ( 1 / \beta ) \leq o ( \log^{1/2} (n) )$.
        \item (Low Discrepancy) $\odisc_C \left( \Psi_i \right) \leq n^{-2} $  with $C \leq n_i / \poly \log (n)$.
    \end{itemize}
\end{repdefinition}
For such $f$, we prove the following simulation theorem.
\begin{theorem} \label{thm:simulation}
    Suppose for any $\vec{S}$, there exists a dynamic data structure for \pref{prob:gmp} with $f$ with update time $t_u = \poly \log (n)$, average query time $\E_{Q,T} \left[\tau_{tot} \right] = t_{tot} \leq o \left( \frac{\log^{3/2} (n)}{( \log \log n )^2} \right).$     
    Then there exists an epoch $i \in [\ell / 3 , 5 \ell / 6  ]$, an event $E_i^1$ and a query set $\cQ_i \subset \cQ$ that depend only on $\vec{S}^{-i}, T_{-i}$ with $\Pr_{ \vec{S}^{-i}, T_{-i} } [ E_i^1 ] \geq \Omega ( g_i )$ such that
    \begin{align*}
        & |\cQ_i| \geq g_i m / 2 \\
        & \forall Q \in \cQ_i,~~ \Pr_{S^i_Q} \left[ \abs{ \E_{T_i} \left[ f ( S_Q, T )  \right] }  \geq \frac{1 - \beta}{1+ \beta} \right] \leq n^{-2},~ \odisc_C ( \Psi_i ) \leq  n^{-2}.
    \end{align*}
    Furthermore, there exists a sub-event $E_i^2 \subset E_i^1$ with $\Pr[ E_i^2 | E_i^1 ] \geq 9/10$, and a query set $\cQ'_i \subset \cQ_i$ with $|\cQ'_i| \geq \left( 1 - \frac{10}{n^2} \right) |\cQ_i| $ depending only on $\vec{S}, T_{-i}$, conditioned on which there exists a one-way communication protocol that obtains  
    \begin{equation*}
        \adv ( G_{\overline{f}_i}, \cD_{\cQ'_i}, \cD_{T_i}, n_i / \poly \log (n) ) \geq n^{-o(1)}.
    \end{equation*}
\end{theorem}

We will use \pref{thm:ly_simulation} as our main ingredient, which we reiterate below.

\begin{reptheorem}{thm:ly_simulation}[\cite{larsen_super-logarithmic_2025}]
    Let $\cP : \cQ \times \cT \to [-1,+1]$ be a data structure problem with weights, and $\cD_{\cQ} \times \cD_{T}$ be a distribution over the queries and inputs such that $\E_{T} \left[ \cP ( Q, T ) \right] = 0$ for all $Q \in \cQ$ and $|\cP ( Q, T )| \geq \beta$ for some $\beta > 0$. If there is a data structure $D$ with pre-initialized memory and a cache for $\cP$ such that $D$ has at most $S$ updated cells, $S_{cac}$ cells in cache, expected query time $t_{tot}$, and expected query time into updated cells $t_q$, then for any $p \in (0,1)$, there exists a one-way communication protocol $G_{\cP_i}$ such that 
    \begin{align*}
        & \adv ( G_{\cP}, \cD_{\cQ}, \cD_{T}, ( 8 p S + S_{cac} ) w ) \\
        & \geq \exp \left( - O \left ( \log ( 1/p) ( t_q + \sqrt{ t_{tot} ( t_q \log (1/p) + \log (1 / \beta) )}  + \log (1 / \beta) \right)  \right) - \exp \left( - \Omega( p S ) \right)
    \end{align*}
\end{reptheorem}

\begin{remark}
    Note that the $\beta$ in the simulation is the same $\beta$ in the definition of hard $f$. Recall that we normalize $f_i (S_Q^i , T_i )$ so that the expectation is zero. We leave it to the reader to verify that the balanced condition $\abs{ \E_{T_i} \left[ f_i ( S^i_Q, T_i )  \right] } $, is exactly the condition required to lower bound the normalized absolute value.
\end{remark}

\begin{proofof}{\pref{thm:simulation}}
    Before directly applying \pref{thm:ly_simulation}, we need to filter through a sequence of events and query sets towards the final lower bound.

    First, we prove the following claim, which gives a fixed epoch $i \in [\ell/3, 5\ell/6]$, an event $E_i^1$ and query set $\cQ_i$ from $\vec{S}^{-i}, T_{-i}$ that are ``good."

    \begin{claim} \label{cl:good_event_minus_i}
    There exists an epoch $i \in [ \ell / 3 , 5 \ell / 6  ]$, an event $E_i^1$, and a query set $\cQ_i \subset \cQ$ depending only on $\vec{S}^{-i}$ and $T_{-i}$ such that: 
    \begin{align*}
        & \Pr_{\vec{S}^{-i}, T_{-i}} [ E_i^1 ] \geq \Omega ( g_i ),~~~ |\cQ_i| \geq \frac{g_i m}{2} \\
        & \forall Q \in \cQ_i,~ \odisc_C ( \Psi_i ) \leq  n^{-2},~~~  \Pr_{S^i_Q} \left[ \abs{ \E_{T_i} \left[ f ( S_Q, T )  \right] }  \geq \frac{1 - \beta}{1+ \beta} \right] \leq n^{-2} \\
        & \E_{ Q, T_i | E_i^1 }  \left[ \tau_Q \right] \leq \frac{ 8 \cdot t_{tot} }{g_i^2 \cdot \ell},~~~ \E_{ Q, T_i | E_i^1 } \left[ \tau_{tot} \right] \leq \frac{ 4 \cdot t_{tot} }{g_i^2} 
    \end{align*}    
    \end{claim}
    \begin{proof}
        First, a simple Markov argument shows that there must exist an epoch $i \in [ \ell / 3 , 5 \ell / 6  ]$ such that
        \begin{align}
            \E_{ Q, \vec{S}, T }  \left[ \tau_q \right] \leq  2 \cdot t_{tot} / \ell. \label{eq:t_q_condition}
        \end{align}
        Furthermore, due to our assumption on the underlying function $f$, we have that for each $Q \in \cQ$,
        \begin{align}
            \Pr_{S_Q^{-i}, T_{-i}} \left[ \left( \Pr_{S^i_Q} \left[ \abs{ \E_{T_i} \left[ f ( S_Q, T )  \right] }  \geq \frac{1 - \beta}{1+ \beta} \right] \geq n^{-2} \right) \vee \left( \odisc_C ( \Psi_i ) \geq  n^{-2} \right) \right] \leq 1 - g_i. \label{eq:disc_condition}
        \end{align}
        By Markov's inequality, with probability at least $1 - \frac{1 - g_i}{1 - (g_i / 2 )} = \frac{g_i }{2 - g_i} \geq \frac{g_i}{2}$ over $\vec{S}^{-i}$ and $T_{-i}$, there exists $\cQ_i \subset \cQ$ with $|\cQ_i| \geq \frac{g_i m}{2}$ such that for all $Q \in \cQ_i$
        \begin{align*}
            & \Pr_{S^i_Q} \left[ \abs{ \E_{T_i} \left[ f ( S_Q, T )  \right] }  \geq \frac{1 - \beta}{1+ \beta} \right] \leq n^{-2}    \\
            & \odisc_C ( \Psi_i ) \leq  n^{-2}
        \end{align*}
        Denote the event corresponding to this choice of $\vec{S}^{-i}, T_{-i}$ as $E_i^1$. Then as $\Pr[ E_i^1 ] \geq \frac{g_i}{2}$ and $|\cQ_i| \geq \frac{g_i m}{2}$, this implies that
        \begin{align*}
            & \E_{ Q \in \cQ_i , \vec{S^i}, T_i | E^1_i } \left[ \tau_q \right] \leq \left( \frac{2}{g_i} \right)^2 \cdot 2 \cdot t_{tot} / \ell  \\
            & \E_{ Q \in \cQ_i , \vec{S^i}, T_i | E^1_i } \left[ \tau_{tot} \right] \leq \left( \frac{2}{g_i} \right)^2 t_{tot}.
        \end{align*}
        which completes the proof of the claim.
    \end{proof}

    Next, we proceed to the second filter, where we condition on $\vec{S}^i$ as well. 
    Conditioned on $E^1_i$, we would like to have a sub-event $E^2_i \subset E^1_i$, and query subset $\cQ'_i$ further conditioned on $\vec{S}^i$.
   
    \begin{claim} \label{cl:good_event}
    There exists an event $E^2_i$, query set $\cQ'_i \subset \cQ_i$ with $|\cQ'_i| \geq \left( 1 - \frac{10}{n^2} \right) |\cQ_i | $ (depending only on $\vec{S}$ and $T_{-i}$) such that  
    \begin{align*}
        & \Pr_{\vec{S}, T_{-i}} [ E_i^2 | E_i^1 ] \geq 9/10 \\
        & \forall Q \in \cQ'_i,~ \odisc_C ( \Psi_i ) \leq  n^{-2}, ~~ \abs{ \E_{T_i} \left[ f ( S_Q, T )  \right] }  \leq \frac{1 - \beta}{1+ \beta}  \\
        & \E_{ Q, T_i | E_i^2 }  \left[ \tau_q \right] \leq \frac{ 10 \cdot t_{tot} }{g_i^2 \cdot \ell},   ~~~  \E_{ Q, T_i | E_i^2 } \left[ \tau_{tot} \right] \leq \frac{ 5 \cdot t_{tot} }{g_i^2} 
    \end{align*}
    \end{claim}
    \begin{proof}
       
        Now, we switch our attention to choosing $\vec{S^i}$ and the respective $\cQ'_i \subset \cQ_i$ assuming $E^1_i$.
        Then there exists an event $E^2_i$ depending on the additional $\vec{S^i}$ such that 
        \begin{align*}
            \Pr_{Q \sim \cQ_i} \left [ \abs{ \E_{T_i} \left[ f ( S_Q, T )  \right] }  \leq \frac{1 - \beta}{1+ \beta} \right] \geq \left( 1 - \frac{10}{n^2} \right),
        \end{align*}
        that is, a choice of $\vec{S}^i$ such that at most $\frac{10}{n^2}$ of the resulting $f(S_Q, T)$'s are {\bf not} well-balanced. Then by Markov's inequality,
        $$1 - \Pr_{\vec{S}^i, \vec{S}^{-i}, T_{-i} } [ E^2_i | E^1_i ] \leq \frac{ n^{-2}}{ \frac{10}{n^2} } = 1/10.$$
        Assuming $E^2_i$, $\vec{S}^i$ yields a set $\cQ'_i \subset \cQ_i$ such that for every $Q \in \cQ'_i$,
        \begin{align*}
            & \odisc_C ( \Psi_i ) \leq  n^{-2},~~ \abs{ \E_{T_i} \left[ f ( S_Q, T )  \right] }  \leq \frac{1 - \beta}{1+ \beta} 
        \end{align*}
        Note that the first condition only depends on $\vec{S}^{-i}, T_{-i}$, while the second condition depends on $\vec{S}, T_{-i}$. Then over $Q\in \cQ'_i$, conditioned on $E^2_i$, by a simple Markov argument,
        \begin{align*}
            & \E_{ Q \in \cQ'_i , T_i | E^2_i } \left[ \tau_q \right] \leq \frac{10}{9} \cdot \frac{1}{1 - \frac{10}{n^2}} \left( \frac{2}{g_i} \right)^2 \cdot 2 \cdot t_{tot} / \ell \leq \frac{ 10 \cdot t_{tot} }{g_i^2 \ell} \\
            & \E_{ Q \in \cQ'_i , T_i | E^2_i } \left[ \tau_{tot} \right] \leq \frac{10}{9} \cdot \frac{1}{1 - \frac{10}{n^2}} \cdot \frac{ 4 \cdot t_{tot} }{g_i^2} \leq \frac{ 5 \cdot t_{tot} }{g_i^2}
        \end{align*}
    \end{proof}

    Now we have the required definition of event $E_i^2$ and $\cQ'_i$. We proceed to apply \pref{thm:ly_simulation} conditioned on $E_i^2$ and over the query set $\cQ'_i$. Given \pref{cl:good_event}, we zoom into the event $E^2_i$. Given $E^2_i$, which is determined by $\vec{S}, T_{-i}$, there exists a data structure with updated cells $S \leq t_u n_i$ and cache $S_{cac} \leq \sum_{ i'' < i } t_u n_{i''} = t_u \frac {n_i - 1}{\gamma -1} \leq \frac{2 t_u n_i}{\gamma} \leq \frac{n_i}{(t_u w)^{\Theta(1)}}$, with expected query time $\frac{ 5 \cdot t_{tot} }{g_i^2}$ and expected query time into updated cells $\frac{ 10 \cdot t_{tot} }{g_i^2 \ell}$ over some query set $\cQ'_i \subset \cQ_i$, where the bound on $S_{cac}$ holds from setting $\gamma = (t_u w)^{\Theta(1)}$. \pref{thm:ly_simulation} then implies a one-way communication protocol for $f_i ( S^i_Q, T_i )$ for any $p \in ( 0, 1)$ with advantage 
    \begin{align}
        & \adv ( G_{\overline{f}_i}, \cD_{\cQ'_i}, \cD_{T_i}, ( 8 p S + S_{cac} ) w ) \nonumber \\
        & \geq \exp \left( - O \left ( \log ( 1/p) \left( t_q + \sqrt{ t_{tot} ( t_q \log (1/p) + \log (1 / \beta) )} \right)  + \log (1 / \beta) \right)  \right) - \exp \left( - \Omega( p S ) \right) \label{eq:advantage_bound}
    \end{align}
    We set $p := \frac{1}{(t_u w)^{\Theta(1)}}$. Then observe that $\ell = \Omega( \log_{\gamma} (n) ) = \Omega \left( \frac{\log n}{\log \log n} \right).$ Plugging in the bounds, we get
    \begin{equation} \label{eq:advantage_middle_bound}
        \beta^{-O(1)} \cdot 2^{ - O \left ( \log ( t_u w ) \left( \frac{t_{tot}}{ g_i^2 \log_{\gamma} (n) } + \sqrt{ \frac{t_{tot}^2 \log ( t_u w ) }{g_i^4 \log_{\gamma} (n) }    + \frac{t_{tot} \log (1 / \beta) }{g_i^2} } \right)  \right)} - \exp \left( - \frac{n_i}{(t_u w)^{\Theta(1)} }  \right).
    \end{equation}
    Assuming $t_u = \poly \log (n)$, $w = O ( \log n)$, $n_i \geq \Omega( n^{1/3} )$ due to the epoch being $i \in [\ell/3, 5\ell / 6]$,
    if we assume $\log ( 1 / \beta ) \leq o ( \log^{1/2} (n) )$, 
    and $t_{tot} \leq o \left( \frac{ g_i^2 \cdot \log^{3/2} (n)}{ (\log \log (n))^{2}} \right)$, \pref{eq:advantage_middle_bound} simplifies to 
    \begin{equation*}
        \pref{eq:advantage_middle_bound} \geq \beta^{-O(1)} \cdot n^{- o(1)} \geq n^{-o(1)}.
    \end{equation*}
    Therefore, there exists a message of length $ \frac{n_i}{\poly \log (n)}$ with $n^{-o(1)} = n_i^{-o(1)}$ advantage for $G_{\overline{f}_i}$
\end{proofof}

\subsection{Combinatorial Lower Bound} \label{sec:lower_bound}

Next, we would like to show that no such one-way protocol exists. Before we delve into the proof of impossibility, we briefly summarize the properties we attained in \pref{sec:simulation}. 

\paragraph{Recap of Previous Section} In the previous section, we showed that a ``too-good-to-be-true" dynamic data structure yields a ``too-good-to-be-true" one-way communication protocol, which had the following property: 

\begin{itemize}
    \item $\cQ_i \subset \cQ$ which depends only on $\vec{S}^{-i}$ and $T_{-i}$, with size $\geq g_i m / 2$ such that all $Q \in \cQ_i$ have:
    \begin{align*}
        & \Pr_{S^i_Q} \left[ \abs{ \E_{T_i} \left[ f ( S_Q, T )  \right] }  \geq \frac{1 - \beta}{1+ \beta} \right] \leq n^{-2}    \\
        & \odisc_C ( \Psi_i ) \leq  n^{-2};
    \end{align*}
    \item An event $E_i^2$ with probability $\Pr_{\vec{S},T_{-i}} [E_i^2 | E_i^1 ] \geq 9 / 10$ and a set of good queries $\cQ'_i \subset \cQ_i$ with $|\cQ'_i | \geq \left( 1 - \frac{10}{n^2}\right) |\cQ_i |$ depending only on $\vec{S}, T_{-i}$ where conditioned on $E_i^2$
    \begin{align*}
        & \adv ( G_{\overline{f}_i}, \cD_{\cQ'_i}, \cD_{T_i}, n_i / \poly \log (n) ) \geq n^{-o(1)}.
    \end{align*}
\end{itemize}
We will first pick ``good" $\vec{S}^{-i}, T_{-i}$ satisfying $E_i^1$ as part of public randomness (i.e., we will assume the event $E_i^1$). Note that $\vec{S}^{-i}, T_{-i}$ reveals zero information about $\vec{S}^i$ and $T_i$ due to our hard distribution being independent. Thus, conditioning on the event will not affect the distribution over $\vec{S}^i$ and $T_i$. 

% As we will only denote the event $E_i^2$ in this section, by picking ``good" $\vec{S}^{-i}, T_{-i}$ , we will drop the super-script for brevity and denote simply as $E_i$, which is then fully determined by $\vec{S}$ and $T_{-i}$.

\paragraph{Main Lower Bound}
We would like to show that these properties cannot be satisfied simultaneously. To that end, we prove the following combinatorial black-box theorem, which can be used to derive the contradiction.

    \begin{theorem}
        \label{thm:main_combinatorial}
        If $|M| \leq C = \frac{n_i}{ \poly \log n}$,  
    the maximum advantage Bob can attain, conditioned on an event $E_i$ determined by $\vec{S},T_{-i}$ is
    \begin{equation*}
        \adv ( G_{\overline{f}_i} , \cD_{\cQ'_i}, \cD_{T_i}, C ) \leq n^{-\Omega(1)}
    \end{equation*}
    \end{theorem}

    Before delving into the main proof, we would like to give a brief overview of the proof, and what combinatorial properties are necessary to derive a contradiction. Our main proof strategy roughly follows the technique pioneered in \cite{ko_adaptive_2020} and subsequently \cite{ko_lower_2025}. 
    We show that a ``too-good-to-be-true" one-way message $M$ achieving a large advantage leads to a ``too-good-to-be-true" random process $Z$ which achieves (i) small information on $T_i$; (ii) small information on $S_Q^i$; and (iii) little correlation between $S_Q^i$ and $T_i$, while such $Z$ cannot exist, as the underlying $f$ is hard.

    We define our ``too-good-to-be-true" random process $Z$ as follows (which would be guaranteed by setting $E_i = 1$). Let $\cP$ be a non-overlapping random sequence of numbers in $\cQ_i$ of length $l$, where $l$ is a parameter we will decide in the final part of the proof. Note that $\cQ_i$ is completely determined by $\vec{S}^{-i}, T_{-i}$. Therefore, $\cP$ reveals zero information about $\vec{S}^i$. We denote the number at position $j$ as $\rho_j$. Then we also take a random number $J \in [l]$ uniformly at random. Let $B_Q$ denote Bob's response for query $Q \in [m]$, and a Boolean random variable $G_Q$ whether $Q \in \cQ'_i$. Our $Z$ is then defined as 
    \begin{equation*}
        Z := M, \cP_{< J } , G_{\rho_{ < J }}, B_{\rho_{ < J }}, S_{\rho_{ < J }}^i, J 
    \end{equation*}
    Then we attach $Z$ along with $\rho_J, G_{\rho_{ J }}, B_{\rho_{ J }}$ and conditioning on $E_i = 1$ to derive the contradiction. 

    \subsubsection{Technical Lemmas}
   
    Towards the proof of \pref{thm:main_combinatorial}, we prove the necessary technical claims first.     The main technical deviation from previous works is on how to deal with the information on $S_{\rho_J}^i$. Just as in \cite{ko_adaptive_2020}, $Z$ reveals little information about $S_{\rho_J}^i$ in terms of mutual information or KL divergence.

     \begin{claim}[Low information on $S_{\rho_J}^i$] \label{cl:little_info_s}
         % \begin{equation*}
         %      I ( S_{\rho_J}^i ; Z, \rho_J, G_{\rho_J}, E_i ) \leq \frac{3 (C + \ell)}{m} 
         % \end{equation*}
         % In particular, this implies that 
         \begin{equation*}
             \E_{Z, \rho_J | E_i = 1 } \left[ D ( S_{\rho_J}^i |_{Z, \rho_J , E_i = 1 }  || S_{\rho_J}^i ) \right] \leq \frac{5 (C + l)}{m}  
         \end{equation*}
     \end{claim}
     \begin{proof}
         % We start by observing that 
         % \begin{equation*}
         %      I ( S_{\rho_J}^i ; Z, \rho_J, G_{\rho_J}, E_i  ) \leq I ( S_{\rho_J}^i ; Z, \rho_J, E_i ) + H ( G_{\rho_J} | E_i ) \leq I ( S_{\rho_J}^i ; Z, \rho_J, E_i ) + \frac{1}{n}
         % \end{equation*}
         % where the inequality holds as if $E_i = 1$, $G_{\rho_J}$ is 0 with at most probability $\frac{10}{n^2}$. Otherwise $E_i = 0$ with probability at most $ 1 - \frac{g_i}{10}$

         Suppose we condition on the event $E_i^1$. 
         Then we use the analogous direct sum technique from \cite{ko_adaptive_2020}. First observe that
         \begin{align*}
             I ( S_{\rho_J}^i ; Z, \rho_J,  E_i^2 | E_i^1 = 1 ) & \leq I ( S_{\rho_J}^i ; M, \cP, B_{\rho_{ < J }}, G_{\rho_{ < J }}, S_{\rho_{ < J }}^i, E_i^2,  J | E_i^1 = 1  ) \\
             & = I ( S_{\rho_J}^i ; M,  E_i^2 , B_{\rho_{ < J }} | \cP, S_{\rho_{ < J }}^i , J, E_i^1 = 1 )
         \end{align*}
         For any fixed $J = j$, note that
         \begin{align*}
             & I ( S_{\rho_J}^i ; M, E_i^2, B_{\rho_{ < j }}, G_{\rho_{ < j }}  | \cP, S_{\rho_{ < J }}^i, J = j, E_i^1 = 1) \leq \frac{1}{m - j} I ( S_{ - \rho_{<j}}^i ; M, E_i^2, B_{\rho_{ < j }} G_{\rho_{ < j }}  | S_{\rho_{ < j }}^i , E_i^1 = 1) \\
             & \leq \frac{C + 2 j + 1}{m - j} \leq \frac{C + 2 l + 1}{m - l} \leq \frac{3 (C + l)}{m}.
         \end{align*}
         Then note that since $\Pr[ E_i^2 | E_i^1  ] \geq  9 / 10$, and 
         \begin{equation*}
             \E_{Z, \rho_J, E_i^2 | E_i^1 = 1 } \left[ D ( S_{\rho_J}^i |_{Z, E_i^2}  || S_{\rho_J}^i ) \right] = I ( S_{\rho_J}^i ; Z, \rho_J,  E_i^2 | E_i^1 = 1 ),
         \end{equation*}
         this completes the proof of the claim, by writing $E_i = E_i^1 \wedge E_i^2$.
     \end{proof}

     We will need to translate the bound on $S^i_{\rho_J}$ from \pref{cl:little_info_s} using the following claim.

     \begin{claim} \label{cl:s_l_1}
         \begin{equation*}
             \E_{Z, \rho_J | G_{\rho_J} = 1, E_i = 1 } \left[ \norm{ S_{\rho_J}^i |_{Z, \rho_J , G_{\rho_J} = 1,  E_i = 1 }  - S_{\rho_J}^i }_1 \right] \leq 5 \sqrt{ \frac{(C + l)}{m} } + \frac{200}{n^2}
         \end{equation*}
     \end{claim}
     \begin{proof}
         We will upper bound $\norm{ S_{\rho_J}^i |_{Z, \rho_J , G_{\rho_J} = 1,  E_i = 1 }  - S_{\rho_J}^i }_1$ using the triangle inequality, namely
         \begin{align}
             \norm{ S_{\rho_J}^i |_{Z, \rho_J , G_{\rho_J} = 1,  E_i = 1 }  - S_{\rho_J}^i }_1 & \leq \norm{ S_{\rho_J}^i |_{Z, \rho_J , G_{\rho_J} = 1, E_i = 1 }  - S_{\rho_J}^i |_{Z, \rho_J, E_i = 1 } }_1 \label{eq:s_first_bound} \\
             & + \norm{ S_{\rho_J}^i |_{Z, \rho_J, E_i = 1 } - S_{\rho_J}^i }_1 \label{eq:s_second_bound}
         \end{align}
         First, we bound \pref{eq:s_first_bound}. If we consider any fixed $Z, \rho_J$,
         \begin{equation*}
             \norm{ S_{\rho_J}^i |_{Z, \rho_J , G_{\rho_J} = 1, E_i = 1 }  - S_{\rho_J}^i |_{Z, \rho_J, E_i = 1 } }_1 \leq 2 \cdot \Pr[ G_{\rho_J} = 0 | Z, \rho_J , E_i = 1].
         \end{equation*}
         Then taking expectation over $Z, \rho_J$
         \begin{align*}
             & \E_{Z, \rho_J | G_{\rho_J} = 1, E_i = 1 } \left[ \pref{eq:s_first_bound} \right] \leq 2 \cdot \E_{Z, \rho_J | G_{\rho_J} = 1, E_i = 1 } \left[ \Pr[ G_{\rho_J} = 0 | Z, \rho_J , E_i = 1] \right] \\
             & \leq 2 \cdot \left( \E_{Z, \rho_J | E_i = 1 } \left[ \Pr[ G_{\rho_J} = 0 | Z, \rho_J , E_i = 1] \right] + 2 \Pr[ G_{\rho_J} = 0 | E_i = 1 ] \right) \\
             & = 6 \cdot \Pr[ G_{\rho_J} = 0 | E_i = 1 ] \leq \frac{60}{n^2}.
         \end{align*}
         If we take expectation over \pref{eq:s_second_bound}, 
         \begin{align*}
             & \E_{Z, \rho_J | G_{\rho_J} = 1, E_i = 1 } \left[ \pref{eq:s_second_bound} \right] \leq \E_{Z, \rho_J |  E_i = 1 } \left[ \pref{eq:s_second_bound} \right] + 4 \Pr[ G_{\rho_J} = 0 | E_i = 1 ]  \\
             & \leq \E_{Z, \rho_J | E_i = 1 } \left[ \norm{ S_{\rho_J}^i |_{Z, \rho_J, E_i = 1 }  - S_{\rho_J}^i }_1 \right] + \frac{40}{n^2}.
         \end{align*}
         The bound on $ \E_{Z, \rho_J | E_i = 1 } \left[ \norm{ S_{\rho_J}^i |_{Z, \rho_J, E_i = 1 }  - S_{\rho_J}^i }_1 \right]$ can be attained immediately from \pref{cl:little_info_s} due to Pinsker's inequality (\pref{fact:pinsker}) and Jensen's inequality.
         \begin{equation*}
             \E_{Z, \rho_J | E_i = 1 } \left[ \norm{ S_{\rho_J}^i |_{Z, \rho_J, E_i = 1 }  - S_{\rho_J}^i }_1 \right] \leq 2 \sqrt{ \E_{Z, \rho_J | E_i = 1 } \left[  D ( S_{\rho_J}^i |_{Z, \rho_J, E_i = 1 }|| S_{\rho_J}^i )  \right] } \leq 2 \sqrt{\frac{5 (C + l)}{m}}
         \end{equation*}
         Combining all the bounds, we obtain 
         \begin{equation*}
             \E_{Z, \rho_J | G_{\rho_J} = 1, E_i = 1 } \left[ \norm{ S_{\rho_J}^i |_{Z, \rho_J , G_{\rho_J} = 1,  E_i = 1 }  - S_{\rho_J}^i }_1 \right] \leq 5 \cdot \sqrt{ \frac{ (C + l)}{m} } + \frac{100}{n^2}
         \end{equation*}
         completing the proof of the claim.
     \end{proof}

     We need the following claim as well to translate the advantage of the protocol to one-way discrepancy of the underlying function, 
   
     \begin{claim} \label{cl:one-way_to_rectangle}
     Let $A$ and $B$ be some random variable which satisfies
     \begin{align*}
         & \norm{ S_{\rho_J}^i |_{A=a} - S_{\rho_J}^i }_1 \leq \delta_1 \\
         & \E_{ S_{\rho_J}^i } \left[  \abs{ \E_{T_i |_{A = a, B = b} } \left[  \overline{f}_i ( S_{\rho_J}^i , T_i ) \right] } \right] \leq \delta_2 
     \end{align*} 
     then 
     \begin{equation*}
         \E_{ S_{\rho_J}^i |_{A =a, B = b}  } \left[  \abs{ \E_{T_i |_{A =a, B = b} } \left[ \overline{f}_i ( S_{\rho_J}^i , T_i ) \right] } \right] \leq \frac{ \delta_2 + \delta_1 }{\Pr[ B = b | A =a ]}
     \end{equation*}
     \end{claim}
     \begin{proof}
         As $\abs{ \E_{T_i |_{A =a, B = b} } \left[  \overline{f}_i ( S_{\rho_J}^i , T_i ) \right] } \leq 1$ for any setting of $S_{\rho_J}^i$, 
         \begin{align*}
             \E_{ S_{\rho_J}^i |_{A =a} } \left[  \abs{ \E_{T_i |_{A =a, B = b} } \left[  \overline{f}_i ( S_{\rho_J}^i , T_i ) \right] } \right] &\leq 
             \E_{ S_{\rho_J}^i  } \left[  \abs{ \E_{T_i |_{A =a, B = b} } \left[  \overline{f}_i ( S_{\rho_J}^i , T_i ) \right] } \right]
             + \norm{S_{\rho_J}^i |_{ A =a} - S_{\rho_J}^i}_1 \\ 
            & \leq \delta_2 +  \delta_1
         \end{align*}
         Finally, since $\Pr[ B = b | A =a ] > 0$, and $\abs{ \E_{T_i |_{A =a, B = b} } \left[ \overline{f}_i ( S_{\rho_J}^i , T_i ) \right] } \geq 0$, due to a simple Markov argument, we get 
         \begin{align*}
             \E_{ S_{\rho_J}^i |_{A =a, B = b}  } \left[  \abs{ \E_{T_i |_{A =a, B = b} } \left[ \overline{f}_i ( S_{\rho_J}^i , T_i ) \right] } \right] \leq \frac{\delta_2 + \delta_1 }{\Pr[ B = b | A =a ]} 
         \end{align*}
         which completes the proof of the claim.
     \end{proof}

We prove the final claim of this section, where we argue that $Z, B_{\rho_J}, G_{\rho_J}, C_{\rho_J}, E_i$ introduce little correlation between $S_{\rho_J}^i$ and $T_i$.
 
     \begin{claim}[Low Correlation] \label{cl:low_correlation}
         \begin{equation*}
             I ( S_{\rho_J}^i ; T_i | Z, \rho_J, B_{\rho_J}, G_{\rho_J} = 1, E_i = 1) \leq \frac{4C}{l} + \frac{25}{n}
         \end{equation*} 
     \end{claim}
     \begin{proof}

         Towards the proof, we introduce an auxiliary event, $C_{\rho_J}$ for the analysis, which is true if a setting of $Z, \rho_J$ is picked such that
         \begin{equation*}
             1 - \Pr[ G_{\rho_J} | Z, \rho_J, E_i = 1] < \frac{1}{2}.
         \end{equation*}
         Then recall that due to our definition of $E_i^2$ (from \pref{cl:good_event}),
         \begin{equation*}
             \E_{Z, \rho_J | E_i = 1} \left[ 1 - \Pr[ G_{\rho_J} | Z, \rho_J, E_i = 1] \right] \leq \frac{10}{n^2}.
         \end{equation*}
         Due to a simple Markov argument
         \begin{equation}
             1 - \Pr_{Z,\rho_J | E_i = 1} [ C_{\rho_J} | E_i = 1 ] < \frac{20}{n^2}
         \end{equation}

         With $C_{\rho_J}$ defined, the bound on $I ( S_{\rho_J}^i ; T_i | Z, \rho_J, B_{\rho_J}, G_{\rho_J} = 1, E_i = 1)$ can be decomposed as
         \begin{align}
             & I ( S_{\rho_J}^i ; T_i | Z, \rho_J, B_{\rho_J}, G_{\rho_J} = 1, E_i = 1) \leq  I ( S_{\rho_J}^i B_{\rho_J} ; T_i | Z, \rho_J,  G_{\rho_J} = 1, E_i = 1) \nonumber \\
             & = \Pr[ C_{\rho_J} = 1 | G_{\rho_J} = 1, E_i = 1] \cdot I ( S_{\rho_J}^i B_{\rho_J} ; T_i | Z, \rho_J,  G_{\rho_J} = 1, C_{\rho_J} = 1, E_i = 1 ) \label{eq:c_1} \\
               & + \Pr[ C_{\rho_J} = 0 | G_{\rho_J} = 1, E_i = 1] \cdot I ( S_{\rho_J}^i B_{\rho_J} ; T_i | Z, \rho_J,  G_{\rho_J} = 1, C_{\rho_J} = 0, E_i = 1 ) \label{eq:c_0}
         \end{align}
         Then we bound \pref{eq:c_1} and \pref{eq:c_0} separately. For \pref{eq:c_0}, observe that for any fixed $Z, \rho_J$
         \begin{align*}
             & I ( S_{\rho_J}^i B_{\rho_J} ; T_i | Z, \rho_J,  G_{\rho_J} = 1, C_{\rho_J} = 0, E_i = 1 )  \leq n_i 
         \end{align*}
         while from a simple application of Bayes's rule,
         \begin{equation*}
             \Pr[ C_{\rho_J} = 0 | G_{\rho_J} = 1, E_i = 1] = \frac{\Pr[ C_{\rho_J} = 0 \wedge G_{\rho_J} = 1 | E_i = 1] }{\Pr[ G_{\rho_J} = 1 | E_i = 1]  } \leq \frac{ \frac{20}{n^2} \cdot \frac{1}{2} }{1 - \frac{10}{n^2}} \leq \frac{20}{n^2}
         \end{equation*}
         which therefore leads to a bound of 
         \begin{equation}
             \pref{eq:c_0} \leq \frac{20 n}{n^2} = \frac{20}{n}. 
         \end{equation}

         On the other hand, \pref{eq:c_1} can be bounded by
          \begin{align*}
             & I ( S_{\rho_J}^i B_{\rho_J} ; T_i | Z, \rho_J, G_{\rho_J} = 1, C_{\rho_J} = 1, E_i = 1 ) \\
             & \leq \frac{ I ( S_{\rho_J}^i B_{\rho_J} ; T_i | Z, \rho_J, G_{\rho_J}, C_{\rho_J} = 1, E_i = 1 ) }{\Pr[G_{\rho_J} = 1 | Z, \rho_J, C_{\rho_J} = 1, E_i = 1 ]}  \leq 2 \cdot I ( S_{\rho_J}^i B_{\rho_J} ; T_i | Z, \rho_J, G_{\rho_J}, C_{\rho_J} = 1, E_i = 1 ) 
         \end{align*}
         where the bound on $\Pr[G_{\rho_J} = 1 | Z, \rho_J, C_{\rho_J} = 1, E_i = 1 ]$ follows from our definition of $C_{\rho_J}$. Then $I ( S_{\rho_J}^i B_{\rho_J} ; T_i | Z, \rho_J, G_{\rho_J}, C_{\rho_J} = 1, E_i = 1 )$ can be bounded as
         \begin{align*}
             & I ( S_{\rho_J}^i B_{\rho_J} ; T_i | Z, \rho_J, G_{\rho_J}, C_{\rho_J} = 1, E_i = 1 ) \\
             & \leq I ( S_{\rho_J}^i B_{\rho_J} ; T_i | Z, \rho_J, C_{\rho_J} = 1, E_i = 1 ) +  \underbrace{H( G_{\rho_J} | Z, \rho_J, C_{\rho_J} = 1, E_i = 1 )}_{\leq n^{-1}} \\
             & \leq \frac{ I ( S_{\rho_J}^i B_{\rho_J} ; T_i | Z, \rho_J,  E_i = 1 ) }{\Pr[ C_{\rho_J} = 1 | E_i = 1  ]} + \frac{1}{n} \leq 2 \cdot I ( S_{\rho_J}^i B_{\rho_J} ; T_i | Z, \rho_J, E_i = 1 ) + \frac{1}{n}
         \end{align*}
         where the second inequality follows from
         \begin{equation*}
            \Pr[C_{\rho_J} = 1 | E_i = 1] \cdot  I ( S_{\rho_J}^i B_{\rho_J} ; T_i | Z, \rho_J, C_{\rho_J} = 1, E_i = 1 ) \leq I ( S_{\rho_J}^i B_{\rho_J} ; T_i | Z, \rho_J, E_i = 1 ).
         \end{equation*}
         Then we bound the $I ( S_{\rho_J}^i B_{\rho_J} ; T_i | Z, \rho_J, E_i = 1 )$ term. 
         \begin{align*}
             & I ( S_{\rho_J}^i B_{\rho_J} ; T_i | Z, \rho_J, E_i = 1 ) = I ( S_{\rho_J}^i B_{\rho_J} ; T_i | M, \rho_{<J}, \rho_J , B_{\rho_{ < J }}, S_{\rho_{ < J }}^i, J, E_i = 1 ) \\
             & \leq I ( S_{\rho_J}^i B_{\rho_J} ; T_i | M, \cP , B_{\rho_{ < J }}, S_{\rho_{ < J }}^i, J,  E_i = 1 ) \\
             & = \frac{1}{l} \sum_{j=1}^{l} I ( S_{\rho_j}^i B_{\rho_j} ; T_i | M, E_i = 1 , \cP, B_{\rho_{ < j }}, S_{\rho_{ < j }}^i ) \\
             & = \frac{1}{l} I ( S^i_{\cP} B_{\cP} ; T_i | M, E_i = 1 , \cP ) \leq \frac{I (\vec{S}^i, B_{\cP}; T_i | M, E_i= 1 , \cP ) }{l}
         \end{align*}
         where the first inequality holds as $\cP$ is chosen independently at random, thus conditioning on $\rho_{>J}$ increases mutual information from \pref{fact:chainrule1}. Then the $I (\vec{S}^i, B_{\cP} ; T_i | M, E_i= 1, \cP )$ term can be bounded by
         \begin{align*}
             I (\vec{S}^i, B_{\cP}  ; T_i | M, E_i= 1, \cP )  & = \underbrace{I ( \vec{S}^i ; T_i | M, E_i= 1, \cP )}_{ \leq C } + \underbrace{I ( B_{\cP} ; T_i | M, \vec{S}^i, E_i= 1, \cP )}_{=0} \leq C
         \end{align*}
         where $I ( B_{\cP}  ; T_i | M, E_i= 1, \vec{S}^i, \cP ) = 0$ as $M, \vec{S}^i$ fully determine $B_i, G_i$'s for all $i \in [m]$, and 
         \begin{equation*}
             I ( \vec{S}^i ; T_i | M, E_i= 1, \cP ) \leq I ( \vec{S}^i ; M, T_i | E_i= 1, \cP ) =  \underbrace{I ( \vec{S}^i; T_i | E_i= 1, \cP )}_{=0} + I ( \vec{S}^i ; M | T_i, E_i= 1, \cP ) \leq C.
         \end{equation*}
         as $E_i$ is determined by $\vec{S}, T_{-i}$, and the distribution over $T_i$ remains unchanged.
         Therefore, we get that
         \begin{equation*}
             I ( S_{\rho_J}^i B_{\rho_J} ; T_i | Z, \rho_J, G_{\rho_J}, C_{\rho_J} = 1, E_i = 1 ) \leq \frac{2 C}{l} + \frac{1}{n}
         \end{equation*}
         which then implies the bound on \pref{eq:c_1} as
         \begin{equation*}
             \pref{eq:c_1} \leq 2 \left( \frac{2 C}{l} + \frac{1}{n} \right) = \frac{4C}{l} + \frac{2}{n}.
         \end{equation*}
         This completes the proof of the claim as 
         \begin{align*}
             I ( S_{\rho_J}^i B_{\rho_J} ; T_i | Z, \rho_J, G_{\rho_J} = 1,  E_i = 1 ) \leq \frac{4C}{l} + \frac{2}{n} + \frac{20}{n} \leq \frac{4C}{l} + \frac{25}{n}
         \end{align*}
     \end{proof}

     \subsubsection{Proof of \pref{thm:main_combinatorial}}

     With all the technical ingredients in place, we are now ready to prove \pref{thm:main_combinatorial}, which then contradicts \pref{thm:simulation}.
     %%% REVIEW CURSOR
    
     \begin{reptheorem}{thm:main_combinatorial} 
     If $|M| \leq C = \frac{n_i}{ \poly \log n}$,  
     the maximum advantage Bob can attain, conditioned on an event $E_i$ determined by $\vec{S},T_{-i}$ is
      \begin{equation*}
         \adv ( G_{\overline{f}_i} , \cD_{\cQ'_i}, \cD_{T_i}, C ) \leq n^{-\Omega(1)}
     \end{equation*}
     \end{reptheorem}    
     \begin{proof}
         Suppose for every setting of $\vec{S}$, conditioned on the event $E_i$, there exists a one-way message $M$ and guessed value $B_Q$ (depending on $M$ and $Q$) such that 
         \begin{equation*}
         \E_{M, Q \in \cQ'_i} \left[  \abs{ \E_{T_i } \left[ B_Q \cdot \overline{f}_i (S_Q, T_i) \right] }  \right] \geq n^{-o(1)}
         \end{equation*}
         which then implies 
         \begin{equation*}
             \E_{\vec{S}, M, Q \in \cQ'_i | E_i } \left[  \abs{ \E_{T_i } \left[ B_Q \cdot \overline{f}_i (S_Q, T_i) \right] }  \right] \geq n^{-o(1)}.
         \end{equation*}
         As our setting of random process $Z$ is fully determined by the above $\vec{S}, M, Q$ (modulo independent randomness),  
         \begin{equation} \label{eq:contradiction}
             \E_{Z, B_{\rho_J}, S_{\rho_J}^i | E_i, G_{\rho_J}, \rho_J } \left[  \abs{ \E_{T_i } \left[ B_{\rho_J} \cdot \overline{f}_i (S_{\rho_J}^i, T_i) \right] }  \right] \geq n^{-o(1)}.
         \end{equation} 

         For the rest of the proof, we will show that \pref{eq:contradiction} cannot be the case. Without loss of generality, we will assume the following for any fixed $\rho_J$ and $Z = z$
         \begin{equation} \label{eq:extra_condition} 
             \Pr[ B_Q | Z = z, \rho_J = \rho_j, G_{\rho_J}, E_i ] \geq n^{-0.1}
         \end{equation}
         by incurring a cost of only $n^{-0.1}$ in the advantage due to the following observation: With probability $n^{-0.1}$, flip the guessed value. The loss in the advantage then is at most $n^{-0.1}$, which is polynomially small (compared to $n^{-o(1)}$).

         Recall that we can write the advantage over random guessing as 
         \begin{equation} \label{eq:initial_advantage}
             \E_{Z, \rho_J,B_{\rho_J}, S_{\rho_J}^i | G_{\rho_J}, E_i } \left[  \abs{ \E_{T_i | Z, B_{\rho_J}, \rho_J, S_{\rho_J}^i, G_{\rho_J}, E_i } \left[ B_{\rho_J} \cdot \overline{f}_i (S_{\rho_J}^i, T_i) \right] }  \right].
         \end{equation}
         Then we can upper bound \pref{eq:initial_advantage} as
         \begin{align} 
             & \pref{eq:initial_advantage} \leq \E_{Z, \rho_J, B_{\rho_J}, S_{\rho_J}^i | G_{\rho_J}, E_i  } \left[  \abs{ \E_{T_i | Z, B_{\rho_J}, \rho_J, S_{\rho_J}^i, G_{\rho_J}, E_i } \left[ \overline{f}_i (S_{\rho_J}^i, T_i) \right] }  \right] \nonumber \\
             & = \E_{Z, \rho_J, B_{\rho_J} | G_{\rho_J}, E_i   } \left[ \E_{ S_{\rho_J}^i | Z, B_{\rho_J}, \rho_J,  G_{\rho_J}, E_i } \left[ \abs{ \E_{T_i | S_{\rho_J}^i, Z, B_{\rho_J}, \rho_J,  G_{\rho_J}, E_i } \left[ \overline{f}_i (S_{\rho_J}^i, T_i) \right] } \right] \right] \nonumber \\
             & \leq \E_{Z, \rho_J, B_{\rho_J} | G_{\rho_J}, E_i   } \left[ \E_{ S_{\rho_J}^i | Z, B_{\rho_J}, \rho_J,  G_{\rho_J}, E_i } \left[  \abs{ \E_{T_i | Z, B_{\rho_J}, \rho_J,  G_{\rho_J}, E_i    } \left[ \overline{f}_i (S_{\rho_J}^i, T_i) \right] } \right] \right] \label{eq:discrepancy_intermediate} \\
             & + \E_{Z, \rho_J, B_{\rho_J} | G_{\rho_J}, E_i   } \left[ \E_{ S_{\rho_J}^i | Z, B_{\rho_J}, \rho_J,  G_{\rho_J}, E_i } \left[ \norm{ T_i |_{Z, B_{\rho_J}, \rho_J, S_{\rho_J}^i, G_{\rho_J}, E_i   } -  T_i |_{Z, B_{\rho_J}, \rho_J,  G_{\rho_J}, E_i  } }_1 \right] \right] \label{eq:l_1_difference_intermediate} 
         \end{align}
        Now if we focus on the \pref{eq:l_1_difference_intermediate} term, we get
         \begin{align}
             \pref{eq:l_1_difference_intermediate} & = \E_{Z, \rho_J, B_{\rho_J} | G_{\rho_J}, E_i   } \left[ \norm{S_{\rho_J}^i |_{Z, B_{\rho_J}, \rho_J,  G_{\rho_J}, E_i } \times T_i |_{Z, B_{\rho_J}, \rho_J,  G_{\rho_J}, E_i } - (S_{\rho_J}^i, T_i) |_{Z, B_{\rho_J}, \rho_J,  G_{\rho_J}, E_i } }_1 \right] \nonumber \\
             & \leq \sqrt{ 2 \cdot I ( S_{\rho_J}^i ; T_i | Z, \rho_J, B_{\rho_J}, G_{\rho_J} = 1, E_i = 1) } \leq  20 \sqrt{ \frac{C}{l} + \frac{1}{n} } \label{eq:l_1-difference}
         \end{align}
         due to \pref{cl:low_correlation} and Pinsker's inequality (\pref{fact:pinsker}).

         Next, we proceed to bound \pref{eq:discrepancy_intermediate} using \pref{cl:little_info_s}, and \pref{cl:one-way_to_rectangle}.
         % \begin{align}
         %      & \tilde{H}_\infty ( T_i | Z, \rho_J, B_{\rho_{ J }} , G_{\rho_J}, E_i ) \geq  H_\infty ( T_i ) - C \label{eq:min-entropy_condition} \\
         %      & \E_{Z, \rho_J | E_i } \left[ D ( S_{\rho_J}^i |_{Z, \rho_J | E_i }  || S_{\rho_J}^i ) \right] \leq \frac{6 (C + \ell)}{m}    \label{eq:divergence}  
         % \end{align}
         We first consider a fixed $Z = z$ $\rho_J$ and $B_{\rho_J} = b$ conditioned on $E_i, G_{\rho_J}$. That is,
         \begin{align}
             \E_{ S_{\rho_J}^i |  Z = z, B_{\rho_J} = b, \rho_J,  G_{\rho_J}, E_i } \left[  \abs{ \E_{T_i | Z = z, B_{\rho_J} = b, \rho_J,  G_{\rho_J}, E_i    } \left[ \overline{f}_i (S_{\rho_J}^i, T_i) \right] } \right] \label{eq:fixed_z}
         \end{align}
         Now we are ready to use \pref{cl:one-way_to_rectangle} to bound \pref{eq:fixed_z}. \pref{cl:one-way_to_rectangle} implies that for any fixed setting of $Z, \rho_J$ and $B_{\rho_J} = b$, 
         \begin{align*}
               & \E_{ S_{\rho_J}^i | Z = z , \rho_J, B_{\rho_J} = b ,  G_{\rho_J}, E_i } \left[ \abs{ \E_{T_i | Z = z , \rho_J, B_{\rho_J} = b ,  G_{\rho_J}, E_i  } \left[ \overline{f}_i (S_{\rho_J}^i, T_i) \right] }  \right] \\
               & \leq \frac{  \norm{ S_{\rho_J}^i |_{Z=z, \rho_J, E_i, G_{\rho_J}} - S_{\rho_J}^i }_1 + \E_{ S_{\rho_J}^i} \left[ \abs{ \E_{T_i | Z = z , \rho_J, B_{\rho_J} = b ,  G_{\rho_J}, E_i } \left[ \overline{f}_i (S_{\rho_J}^i, T_i) \right] }  \right] }{\Pr[ B_{\rho_J} = b | Z = z, \rho_J, G_{\rho_J}, E_i ]} 
         \end{align*}
         
         Now we would like to take expectation over $Z$ and $B_{\rho_J}$ and $\rho_J$ conditioned on the event $G_{\rho_J}$ and $ E_i $.
         % \pref{cl:bayes} implies that if $\Pr[ Z = z, \rho_J, B_{\rho_J} = b | G_{\rho_J} = 1, E_i ] > 0$
         % \begin{equation} \label{eq:bayes}
         %      \frac{ \Pr[ Z = z, \rho_J, B_{\rho_J} = b | G_{\rho_J} = 1, E_i ] }{\Pr[G_{\rho_J} = 1, B_{\rho_J} = b | Z = z, \rho_J, E_i ]} = \frac{\Pr[Z = z, \rho_J | E_i  ]}{\Pr[ G_{\rho_J} = 1 | E_i ]}
         % \end{equation}
         Decomposing the upper bound for \pref{eq:discrepancy_intermediate}, we obtain
         \begin{align}
             \pref{eq:discrepancy_intermediate} & \leq \E_{Z, \rho_J, B_{\rho_J} | G_{\rho_J}, E_i   } \left[ \frac{ \norm{ S_{\rho_J}^i |_{Z=z, \rho_J, E_i, G_{\rho_J}} - S_{\rho_J}^i }_1 }{\Pr[ B_{\rho_J} = b | Z = z, \rho_J, G_{\rho_J} = 1, E_i ]}  \right] \label{eq:S_div} \\
             & + \E_{Z, \rho_J, B_{\rho_J} | G_{\rho_J}, E_i   } \left[ \frac{ \E_{ S_{\rho_J}^i} \left[ \abs{ \E_{T_i | Z = z , \rho_J, B_{\rho_J} = b ,  G_{\rho_J}, E_i } \left[ \overline{f}_i (S_{\rho_J}^i, T_i) \right] }  \right] }{\Pr[ B_{\rho_J} = b | Z = z, \rho_J, G_{\rho_J} = 1, E_i ]}  \right] \label{eq:hardness_f} 
         \end{align}

         The bound on \pref{eq:S_div} follows immediately from \pref{cl:s_l_1}.
         \begin{align*}
             \pref{eq:S_div} & \leq 5 \cdot \sqrt{ \frac{ (C + l)}{m} } + \frac{100}{n^2} 
         \end{align*}
         While for \pref{eq:hardness_f}, rewriting the expectation and applying \pref{eq:extra_condition}, we get
         \begin{align*}
             \pref{eq:hardness_f} & \leq n^{0.1} \cdot \E_{Z, \rho_J, B_{\rho_J} | G_{\rho_J}, E_i   } \left[ \E_{ S_{\rho_J}^i} \left[ \abs{ \E_{T_i | Z = z , \rho_J, B_{\rho_J} ,  G_{\rho_J}, E_i } \left[ \overline{f}_i (S_{\rho_J}^i, T_i) \right] }  \right]   \right].
         \end{align*}
         Now we would like to use the assumption about $f_i$ to upper bound:
         \begin{align}
         & \E_{Z, \rho_J, B_{\rho_J} | G_{\rho_J}, E_i   } \left[ \E_{ S_{\rho_J}^i} \left[ \abs{ \E_{T_i | Z = z , \rho_J, B_{\rho_J} ,  G_{\rho_J}, E_i } \left[ \overline{f}_i (S_{\rho_J}^i, T_i) \right] }  \right]   \right] \nonumber \\
         & = \E_{ S_{\rho_J}^i} \left[ \E_{Z, \rho_J, B_{\rho_J} | G_{\rho_J}, E_i   } \left[  \abs{ \E_{T_i | Z = z , \rho_J, B_{\rho_J} ,  G_{\rho_J}, E_i } \left[ \overline{f}_i (S_{\rho_J}^i, T_i) \right] }  \right]   \right] 
         \end{align}
         Recall that due to our assumption on $f_i$, in particular, the low discrepancy condition, we have that for any $\vec{S}, T_{-i}$, $Q \in \cQ'$, message $M$ which is at most $C$ bits,
         \begin{equation} \label{eq:low_discrepancy}
             \E_{\varsigma_{\rho_J}^i \sim \cU} \left[  \E_{M | \vec{S}, T_{-i}} \left[ \abs{ \E_{T_i | M}  \left[ \Psi_i ( \varsigma_{\rho_J}^i, T_i ) \right] }  \right] \right] \leq n^{-2}.
         \end{equation}
         Recall that
         \begin{equation*}
         Z := M, \cP_{< J } , G_{\rho_{ < J }}, B_{\rho_{ < J }}, S_{\rho_{ < J }}^i, J. 
         \end{equation*}
         \pref{eq:low_discrepancy} then implies (by taking the expectation over $\vec{S^i}$), as $M$ and $\vec{S^i}$ fully determine all the random variables, that for any fixed $\varsigma_{\rho_J}^i$ (chosen independently, thus no relationship to $\vec{S^i}$)
         \begin{align*}
             & \E_{Z, \rho_J, B_{\rho_J} | G_{\rho_J}, E_i   } \left[  \abs{ \E_{T_i | Z = z , \rho_J, B_{\rho_J} = b,  G_{\rho_J}, E_i } \left[ \overline{f}_i ( \varsigma_{\rho_J}^i, T_i) \right] }  \right]  \\
             & =  \E_{M, \cP_{< J } , G_{\rho_{ < J }}, B_{\rho_{ < J }}, S_{\rho_{ < J }}^i, J, \rho_J, B_{\rho_J} | G_{\rho_J}, E_i   } \left[  \abs{ \E_{T_i | Z = z , \rho_J, B_{\rho_J} = b,  G_{\rho_J}, E_i } \left[ \overline{f}_i ( \varsigma_{\rho_J}^i , T_i) \right] }  \right] \\
             & \leq \E_{M, \vec{S^i}, J, \rho_J | G_{\rho_J}, E_i   } \left[ \abs{ \E_{T_i | M, \vec{S^i}, J, \rho_J,  G_{\rho_J}, E_i } \left[ \overline{f}_i ( \varsigma_{\rho_J}^i , T_i) \right] }  \right] \\
             & = \E_{\vec{S^i}, J, \rho_J | G_{\rho_J}, E_i } \E_{M |\vec{S^i}, J, \rho_J,  G_{\rho_J}, E_i } \left[ \abs{ \E_{T_i | M, \vec{S^i}, J, \rho_J,  G_{\rho_J}, E_i } \left[ \overline{f}_i ( \varsigma_{\rho_J}^i , T_i) \right] }  \right].
         \end{align*} 
         We remark that $\vec{S^i}, J, \rho_J,  G_{\rho_J}, E_i$ are all independent of $T_i$, therefore $M |\vec{S^i}, J, \rho_J,  G_{\rho_J}, E_i$ is a $C$-bit message about $T_i$. Then taking expectation over $\varsigma_{\rho_J}^i \sim \cU$, we get
         \begin{align*}
             & \E_{\varsigma_{\rho_J}^i \sim \cU} \left[ \E_{\vec{S^i}, J, \rho_J | G_{\rho_J}, E_i } \E_{M |\vec{S^i}, J, \rho_J,  G_{\rho_J}, E_i } \left[ \abs{ \E_{T_i | M, \vec{S^i}, J, \rho_J,  G_{\rho_J}, E_i } \left[ \overline{f}_i ( \varsigma_{\rho_J}^i , T_i) \right] }  \right] \right] \\
             & =  \E_{\vec{S^i}, J, \rho_J | G_{\rho_J}, E_i } \left[ \underbrace{ \E_{\varsigma_{\rho_J}^i \sim \cU} \left[ \E_{M |\vec{S^i}, J, \rho_J,  G_{\rho_J}, E_i } \left[ \abs{ \E_{T_i | M, \vec{S^i}, J, \rho_J,  G_{\rho_J}, E_i } \left[ \overline{f}_i ( \varsigma_{\rho_J}^i , T_i) \right] }  \right] \right]}_{\leq n^{-2} \mbox{ via \pref{eq:low_discrepancy}}} \right] \leq n^{-2} 
         \end{align*}

         This in turn implies a bound on \pref{eq:hardness_f},
         \begin{align*}
             \pref{eq:hardness_f} & \leq n^{0.1} \cdot n^{-2} = n^{-1.9}.
         \end{align*}
         Finally we get the bound on \pref{eq:discrepancy_intermediate} assuming $\frac{ (C + l)}{m}, \frac{C}{l} = n^{- \Omega(1)}$, which is true as $m = \Omega(n^{0.99} )$ and by choosing $l = n^{0.9}$, as $n_i \leq n^{0.8}$ (since the number of blocks $n_i \leq n/k$ and $i \leq 5 \ell / 6$). This results in 
         \begin{equation*}
             \pref{eq:discrepancy_intermediate} \leq n^{-1.9} + 5 \cdot \sqrt{ \frac{ (C + l)}{m} } + \frac{100}{n^2} \leq n^{- \Omega(1)}
         \end{equation*}
         which then implies our main bound on \pref{eq:initial_advantage} as 
         \begin{equation*}
             \pref{eq:initial_advantage} \leq n^{- \Omega(1)} + 20 \sqrt{ \frac{C}{l} + \frac{1}{n} } \leq  n^{- \Omega(1)}.
         \end{equation*}
         a contradiction.
         
     \end{proof}

     \section{Applications}  \label{sec:multiphase_application}

In this section, we list some applications of our black-box theorem. Note that we only need to show that some particular function $f$ is ``hard." We introduce some additional preliminaries towards that goal.

\subsection{Preliminaries for Min-Entropy}

We will need the following properties of min-entropy for the proof of one-way communication lower bounds.

\begin{definition} We define the renyi entropy $H_2 (A)$ and min-entropy $H_\infty (A)$ as
    \begin{align*}
        H_2 (A) := - \log \left( \sum_{a} \Pr[ A = a]^2 \right) \\
        H_\infty (A) := - \log \left( \max_{a} \Pr[ A = a]  \right)
    \end{align*}
\end{definition}

\begin{fact}[Renyi Entropy] \label{fact:renyi}
Let $A$ be a random variable. Then
    \begin{equation*}
        H ( A ) \geq H_2 (A) \geq H_\infty (A) 
    \end{equation*}
 In particular, for any fixed $b$ we have
\begin{equation*}
H_2 ( A| B = b ) \geq H_\infty ( A | B = b)    
\end{equation*}
\end{fact}

We use the following lemma on ``average'' min-entropy.

\begin{definition}[Average Min-Entropy]
    \begin{equation*}
        \widetilde{H}_\infty ( A | B ) = - \log \left( \E_{b \sim B} \left[ \max_{a} \Pr [ A = a | B = b ] \right] \right) = - \log \left( \E_{b \sim B} \left[ 2^{- H_\infty (A | B = b ) } \right] \right)    
    \end{equation*}
\end{definition}

\begin{lemma}[Lemma 2.2 of \cite{dodis_fuzzy_2008}] \label{lem:dors}
    Let $A,B$ be random variables. Then if $B$ has at most $2^{\lambda}$ possible values, then 
    \begin{equation*}
        \widetilde{H}_\infty ( A | B ) \geq \widetilde{H}_\infty (A, B ) - \lambda \geq H_\infty ( A ) - \lambda.
    \end{equation*}
\end{lemma}

\begin{claim} \label{cl:conditioning_infty}
    \begin{equation*}
        \widetilde{H}_\infty ( A | B,C ) \leq \widetilde{H}_\infty ( A | B  )
    \end{equation*}
\end{claim}
\begin{proof}
    We first proceed with showing the following inequality,
    \begin{equation} \label{eq:conditioning_infty}
        \E_{c \sim C |_{B = b} } \left[ \max_{a} \Pr[ A = a | B = b, C =c ] \right] \geq \max_{a} \Pr[ A = a | B = b ].
    \end{equation}
    Let $a^* := \argmax_{a} \Pr[ A = a | B = b ] $. Then
    \begin{align*}
        \Pr[ A = a^* | B = b ] = \E_{c \sim C |_{B = b} } \left[ \Pr[ A = a^* | B = b, C =c ] \right] \leq \E_{c \sim C |_{B = b} } \left[ \max_{a} \Pr[ A = a | B = b, C =c ] \right] 
    \end{align*}

    With \pref{eq:conditioning_infty} established, we are ready to prove the claim. Recall that 
    \begin{align*}
        \widetilde{H}_\infty ( A | B,C )  := - \log \left( \E_{b,c \sim B,C } \left[ \max_{a} \Pr[ A = a | B = b, C =c ] \right] \right)
    \end{align*}
    Taking expectation over $B$ on both side of \pref{eq:conditioning_infty}, we have
    \begin{align*}
        \E_{b,c \sim B,C } \left[ \max_{a} \Pr[ A = a | B = b, C =c ] \right] \geq \E_{b \sim B } \left[ \max_{a} \Pr[ A = a | B = b ] \right]
    \end{align*}
    Therefore, we get
    \begin{align*}
        \widetilde{H}_\infty ( A | B,C )  & = - \log \left( \E_{b,c \sim B,C } \left[ \max_{a} \Pr[ A = a | B = b, C =c ] \right] \right) \\
        & \leq - \log \left( \E_{b \sim B } \left[ \max_{a} \Pr[ A = a | B = b ] \right] \right) = \widetilde{H}_\infty ( A | B ) 
    \end{align*}

\end{proof}

We can also bound the support size of the distribution using the min-entropy.

\begin{claim} \label{cl:min_entropy_supp}
    Let the distribution on a random variable $X \in \{0,1\}^n$ conditioned on some random variable $M = m$ be supported over the set $T_m \subset \{0,1\}^n$. Then $|T_m| \geq 2^{ H_\infty ( X | M = m) } $  
\end{claim}
\begin{proof}
    Recall that
    \begin{equation*}
        H_\infty ( X | M = m) := - \log \max_{x} \left(  \Pr[ X = x | M = m] \right). 
    \end{equation*}
    Therefore, we have
    \begin{equation*}
        2^{- H_\infty ( X | M = m) } \cdot |T_m| \geq \sum_{x \in T_m} \Pr[ X = x | M = m ] = 1,
    \end{equation*}
    which completes the proof of the claim. 
\end{proof}

The main technical issue with using min-entropy instead of KL-divergence is that {\bf chain-rule does not hold for min-entropy}. Nevertheless, we have the following strong chain-rule if we are willing to ``spoil'' a few bits.

\begin{theorem}[\cite{skorski_strong_2019}] \label{thm:strong_chain_rule_min}
    Let $\cX$ be a fixed alphabet, and $X = (X_1, \ldots, X_t)$ be a sequence of (possibly correlated) random variables each over $\cX$, equipped with a distribution $\mu$. Then for any $\eps \in (0,1)$ and $\delta > 0$, there exists a collection $\cB$ of disjoint sets on $\cX^t$ such that 
    \begin{itemize}
        \item $\cB$ can be indexed by a small number of bits, namely
        $$ \log |\cB| = t \cdot O \left( \log \log |\cX| + \log \log \eps^{-1} + \log ( t / \delta ) \right)  $$
        \item $\cB$ almost covers the domain
        $$ \sum_{B \in \cB} \mu (B) \geq 1 - \eps$$
        \item Conditioned on $\cB$, block distributions $X_{i} | X_{<i}$ are nearly flat. That is
        $$ \forall x, x' \in B,~~ 2^{-O(\delta)} \leq \frac{\mu(x_i | x_{<i} )}{\mu(x'_i | x'_{<i} ) } \leq 2^{O(\delta)} $$
       
        \item For every $B \in \cB$, for every index $i \in [t]$, and for every set $I \subset [i-1]$, we have
        \begin{enumerate}
            \item The chain-rule for min-entropy 
            $$ H_\infty ( X_i | X_I, B ) + H_\infty ( X_I | B ) = H_\infty ( X_i, X_I | B ) \pm O( \delta) $$
            \item The average and worst-case min-entropy almost match
            $$ \widetilde{H}_\infty ( X_i | X_I, B ) = H_\infty ( X_i | X_I, B ) \pm O( \delta ) $$
        \end{enumerate}
    \end{itemize}
\end{theorem}

The corollary of \pref{thm:strong_chain_rule_min} that we will need is the following

\begin{corollary}
    For any fixed $B \in \cB$
    $$ \sum_{i = 1}^t \widetilde{H}_\infty ( X_i | X_{<i}, B ) = H_\infty ( X | B ) \pm t \cdot O ( \delta ) .$$
    In other words, if we take a uniform distribution over $[t]$,
    $$ \E_{i \in [t]} \left[ \widetilde{H}_\infty ( X_i | X_{<i}, B ) \right] = \frac{H_\infty ( X | B )}{t} \pm O ( \delta ) $$
\end{corollary}
\begin{proof}
    We prove by induction on the number of summands. Consider the base case. If there is a single summand, 
    $$ \widetilde{H}_\infty ( X_1 | B ) = H_\infty ( X_1 | B ) $$
    thus this is trivially true. Now assume (as inductive hypothesis) that
    $$ \sum_{i = 1}^{j-1} \widetilde{H}_\infty ( X_i | X_{<i}, B ) = H_\infty ( X_1, \ldots, X_{j-1} | B ) \pm (j-1) \cdot O ( \delta ).$$
    Then for $H_\infty ( X_1, \ldots, X_j | B )$, \pref{thm:strong_chain_rule_min} implies that 
    \begin{align*}
        H_\infty ( X_1, \ldots, X_j | B ) & = \widetilde{H}_\infty ( X_j | X_{<j}, B ) + H_\infty (X_{<j}| B) \pm O( \delta ) \\
        & = \widetilde{H}_\infty ( X_j | X_{<j}, B ) +  \sum_{i = 1}^{j-1} \widetilde{H}_\infty ( X_i | X_{<i}, B ) \pm j \cdot O( \delta ) 
    \end{align*}
    where the second equality holds by the inductive hypothesis. This then completes the proof.
\end{proof}

For our application, we will take $\eps = n^{-\Omega(1)}$. By setting $\eps$ as such, even if we condition on being inside the support of $\cB$, the advantage of the one-way protocol is not affected. 

We will use the following corollary/variant of KKL Theorem. We follow the standard notations from \cite{odonnell_analysis_2014}, to which we refer the reader for Fourier Analysis on Boolean functions.
\begin{lemma}[KKL-Theorem] \label{lem:kkl} 
    Let $T \subset \{ 0 , 1 \}^n$, and $\delta \in [0,1]$. Let $f$ be a function supported on $T$ such that for all $x \in T$, $\abs{f(x)} \leq 1+ \eps $. Suppose we write $\widetilde{f} (S) := \frac{\widehat{f} (S) 2^n}{|T|}$. Then
    \begin{equation*}
        \sum_{S \subset [n] } \delta^{|S|} \widetilde{f} (S)^2 \leq \left( 1+ \eps\right)^{2} \left( \frac{2^{n}}{|T|} \right)^{2\delta}
    \end{equation*}
\end{lemma}
\begin{proof}
    Follows from the usual proof of KKL-Theorem using Hypercontractivity. Let $p = 1 + \delta$, and $\rho = \sqrt{p - 1}$. Then
    \begin{align*}
        & \norm{f}_p^2 = \left( \sum_{x \in T } f(x)^p \right)^{2/p} \leq \left( (1+\eps)^p \cdot \frac{|T|}{2^n} \right)^{2/p} = ( 1+ \eps)^2 \cdot \left( \frac{|T|}{2^n} \right)^{2/p}\\
        & \norm{ T_{\rho} (f) }_2^2 = \sum_{S \subset [n] } \delta^{|S|} \widehat{f} (S)^2
    \end{align*}
    Dividing both sides by $( |T| / 2^n )^2 $, and using the well-known Hypercontractivity which states 
    $\norm{ T_{\rho} (f) }_2^2 \leq \norm{f}_p^2$, 
    we obtain
    \begin{equation*}
        \sum_{S \subset [n] } \delta^{|S|}  \widetilde{f} (S)^2  \leq \left( \frac{|T|}{2^n} \right)^{ - 2 \delta } 
    \end{equation*}
    which completes the proof.
\end{proof}

\subsection{Super-Logarithmic Lower Bound for the Multiphase Problem}

We simply need to show that inner product (mod 2) is ``hard" as in \pref{def:hard}.

\begin{lemma} \label{lem:ip_hard}
    $\IP$ satisfies \pref{def:hard}.
\end{lemma}
\begin{proof}
    For any fixing of $S^{-i}_Q, T_{-i}$, indeed the function is balanced with high probability. As long as $S^i_Q$ is not all zero string. The function is perfectly balanced. The probability of such event is at most $2^{- n_i} \leq 2^{ - n^{1/4}}$, thereby satisfying the balanced condition.

    It remains to show that $f_i$ has low discrepancy. $f_i$ also has low discrepancy due to the following observation. The corresponding $\Psi_i$ is the Hadamard matrix. Then we can use the following well-known lemma to bound the discrepancy.
    \begin{fact}[Lindsey’s Lemma] \label{fact:hadamard}
    Let $H$ be a Hadamard matrix. Let $P$ and $Q$ be distributions. Then 
    \begin{equation*}
        P^{T} H Q \leq \norm{ P }_2 \norm{ Q }_2 \cdot 2^{n/2}
    \end{equation*}
    \end{fact}
    \pref{fact:hadamard} then implies that $T_i$ under the distribution $T_i |_{M = m}$ (the prior on $T_i$ conditioned on the message $M = m$) has
    $$ \odisc_{T_i |_{M = m}} ( \Psi_i ) = v_m^T \cdot  H_{n_i}  \cdot T_i |_{M = m} \leq 2^{n_i/2} \cdot \underbrace{\norm{ v_m }_2}_{ = 2^{-n_i/2}} \cdot \underbrace{\norm{ T_i |_{M = m} }_2}_{\leq 2^{-\frac{H_\infty ( T_i | {M = m} ) }{2}}} \leq 2^{-\frac{H_\infty ( T_i | {M = m} ) }{2}} $$
    where $v_m \in \{ \pm 2^{- n_i} \}^{2^{n_i}}$ is a vector that matches the sign per coordinate of $H_{n_i}  \cdot T_i |_{M = m}$, justifying the norm of such vector.
    
    Since $M$ is of length at most $c = \frac{n_i}{\poly \log (n)} $ and recall that we have defined
    \begin{equation*}
        \odisc_c ( \Psi_i ):= \max_{ M : |M| < c }  \E_{M} \left[ \odisc_{T_i |_M } ( \Psi_i ) \right]
    \end{equation*}   
    
    \begin{align*}
        \E_{M} \left[ \odisc_{T_i | M} ( \Psi_i ) \right] \leq \E_{M} \left[ 2^{-\frac{H_\infty ( T_i | {M = m} ) }{2}} \right] \leq \sqrt{ \E_{M} \left[ 2^{- H_\infty ( T_i | {M = m} ) } \right] } =  2^{- \frac{\widetilde{H}_\infty ( T_i | M )}{2} }
    \end{align*}
    where the second inequality holds from Jensen's inequality (on $\sqrt{x}$) and the equality holds from the definition of average min-entropy. \pref{lem:dors} then implies that $\widetilde{H}_\infty ( T_i | M ) \geq n_i - c$, which in turn implies 
    $$\odisc_c ( \Psi_i )  \leq 2^{- \Omega ( n_i )}, $$ thereby satisfying the low discrepancy if $c = \frac{n_i}{\poly \log (n)} $.     
\end{proof}

\begin{remark}
    Note that the inner product is way stronger than what is necessary for our proof. Nevertheless, this gives the state-of-the-art bound for the Multiphase Problem~\cite{patrascu_mihai_towards_2010}.
\end{remark}

\subsection{A Simple Lifting Theorem}

Lifting technique refers to translating a lower bound for a weaker model to a lower bound for a stronger general model. A classic example of the lifting theorem is Sherstov's Pattern Matrix method~\cite{sherstov_pattern_2011} which translates approximate
degree lower bounds (a structural lower bound for a weaker model) into approximate-rank and communication lower bounds. 

Unfortunately, we cannot make a full translation, as we require the underlying hard distribution to be {\bf a product distribution}, instead of a general distribution. 
In order to completely translate between the some algebraic properties of $\psi$ such as its approximate degree (See \cite{bun_approximate_2022} and references therein), to its communication complexity, we need to consider all possible distribution, instead of just product distribution.

Nevertheless, the following lemma from \cite{sherstov_pattern_2011} can be used to create hard functions. We would like to lift the function using the following composition, $f = \psi \circ g_k^n$ where $g$ (the so-called inner gadget) is
$$ g ( (j,b) , (x_{i+1}, \ldots, x_{i+k}) ) = b \oplus x_{i+j}.$$ 

Now the task in our one-way communication is the following. Alice is given $\varsigma \in \left( [k] \times \{ \pm 1 \} \right)^n$, Bob is given $(X_1, \ldots, X_n ) \in \left( \{ \pm 1 \}^{ k } \right)^n$. Bob sends a one-way message for Alice to compute $$f( \varsigma, X ) = \psi ( g ( \varsigma_1 , X_1 ) , \ldots,  g ( \varsigma_n , X_n ) ).$$ 
We consider the composition, $\psi = \varphi ( \psi_1, \ldots , \psi_\ell )$. The only property we need about $\varphi$ is that when $\psi_{-i}$'s are fixed, $\psi$ value is still undetermined. That is $\psi'_i := \psi |_{\psi_{-i}} $ satisfies the ``hard" condition. For example, it is relatively straightforward to see that if $\psi_i$ is balanced, then $\bigoplus_{i=1}^\ell \psi_i$ is also balanced. For $\bigvee_{i=1}^\ell \psi_i$, we need to have $\psi_i$ biased to say $\min \{ \Pr[ \psi_i = +1 ], \Pr[ \psi_i = -1 ] \} \geq \ell ^{-1}$, the trick introduced in \cite{larsen_super-logarithmic_2025}.

A straightforward connection can be achieved using the following lemma, which can be viewed as a generalization of \pref{fact:hadamard}.
\begin{lemma}[\cite{sherstov_pattern_2011}] \label{lem:pattern_matrix}
    $$ \norm{ \Psi_i } = \sqrt{ 2^{kn} (2k)^{n} } \max_{S \subset [n] } \left( \abs{ \hat{\psi} (S) } \cdot k^{- |S| /2 } \right)$$
\end{lemma}
Directly applying the bound by \cite{sherstov_pattern_2011} and the argument introduced in \pref{sec:multiphase_application} to a uniform distribution over $\varsigma$, we get
\begin{equation*}
    \odisc_c (\Psi_i ) \leq  2^{c/2} \cdot \max_{S \subset [n] } \left( \abs{ \hat{\psi} (S) } \cdot k^{- |S| /2 } \right) 
\end{equation*}
connecting the spectral property of the underlying $\psi$ to its $\odisc$. Note that this is matching (i.e. with \pref{sec:multiphase_application}) when $\psi$ is a parity function. 

\begin{remark}
    In standard communication complexity setting, \pref{lem:pattern_matrix} can be used to obtain a close relation between algebraic properties of $\psi$ and the communication complexity of $f$ (i.e. so-called lifting technique). Unfortunately, this is not the case for our application. We need a lower bound against a product distribution between Alice's input and Bob's input. But the lower bound achieved through connections to algebraic properties of $\psi$ is not necessarily for a product distribution.
\end{remark}

\section{Generalization of 0-XOR} \label{sec:0xor}

As a further illustration of our method, we consider a generalization of a hard function (so-called $0$-XOR) used in \cite{larsen_super-logarithmic_2025}. 

\paragraph{0-XOR Function}

Fix an epoch $i \in [\ell]$. We have the following function. If $n_i$, the size of epoch $i$, is small (say less than $\log^2 (n)$), we have $f_i := 0$. For any other epochs, we consider the following function. We divide the updates to say $ \log \ell $ columns and $R = O ( \log n )$-rows. For fixed column $d$, we select one entry per row, denoted as $\varsigma_i [d]$. Then we take $\AND$ over the columns. Therefore our function in question becomes 
\begin{equation*}
    f_i ( \varsigma_i, T_i) := \AND_{d = 1}^D \left( \ip{ \varsigma_i [d], T_i} \right).
\end{equation*}
Then we take $\OR$ over the epochs $i \in [\ell]$. Note that regardless of which $\varsigma_i$ we select, if $T_i$'s are distributed i.i.d. all bits being $\cB_{1/2}$, $f_i$ is $1$ with probability $2^{-\ell} = \frac{1}{\ell}$. This ensures that our function $\OR ( f_1, \ldots f_{i-1}, f_{i+1}, f_\ell )$ is $1$ with probability at most $1/2$ for any $i \in [\ell/2, \ell]$. In particular, the function is indeed balanced as from \pref{def:hard}.

\subsection{Reinterpretation of \cite{larsen_super-logarithmic_2025}}

\cite{larsen_super-logarithmic_2025} uses a distribution over $\varsigma_i [j]$'s which are highly correlated one another, stemming from the fact that the choice of $\varsigma_i$ derives from a distribution on the underlying graph (the butterfly graph). 

This leads to the distribution over $\varsigma_i$ used in \cite{larsen_super-logarithmic_2025} not having a small $\odisc$, as shown in Section 5.2 of \cite{larsen_super-logarithmic_2025}. This occurs due to $\odisc$-quantity highly depending on the distribution over $\varsigma$, while they need a specific distribution over $\varsigma$ due to their dependence on the underlying Butterfly Graph and the reduction to Graph Connectivity.

Instead, they show $k$-XOR version of the underlying problem is hard. In fact, one could then phrase their technical communication lower bound (on so-called Meta-queries) in Section 5.2 as the following in our notation.

\begin{lemma}[Lemma 5.4 of \cite{larsen_super-logarithmic_2025}]
    There exists a distribution $\cM$ (which is a uniform distribution over its support of size $n^{\Omega(k)}$) on $k$-tuples $(\varsigma_i^1, \ldots, \varsigma_i^k)$ with $k = n / \poly \log (n)$ such that if $|M| < o ( n / \poly \log (n) )$, then the advantage for outputting $\oplus_{j=1}^k f_i ( \varsigma_i^j , T_i )$ is at most $2^{- \Omega ( k \log n )}$.
\end{lemma}

Observe that $\oplus_{j=1}^k f_i ( \varsigma_i^j , T_i )$ is hard (as in \pref{def:hard}) with the parameter $n^{- \Omega(k)}$ instead of $n^{-2}$. Since a too-good-to-be true dynamic data structure (with query time $t_{tot}$ and update time $t_u = \poly \log (n)$) for $0$-XOR implies a data structure with query time $k \cdot t_{tot}$, this implies a one-way message $M$ of length $|M| < o ( n / \poly \log (n) )$ with advantage $n^{-o(k)}$, a contradiction.

Though we omit the details as this is simply a reinterpretation, and do not imply any new result, our argument can further be extended to the settings where we select some subset of $k$-tuples (of size $n^{\Omega(k)}$) from the support of $\cM$, instead of selecting all elements of $\cM$ as the query set.

\subsection{Independent Instances of $0$-XOR} 

Our goal in this section to exhibit another useful example of our method and show that a generalization of $0$-XOR when the queries are independently formed is ``hard."

The main technical challenge here is that other than $\IP$ considered in \pref{lem:ip_hard}, which more-or-less follows from the quasi-random property of the underlying Hadamard matrix (\pref{fact:hadamard}), there does not seem to be a general technique for bounding the $\eps$-advantage one-way protocols with the lower bound of the form $\Omega ( n / \poly \log ( 1 / \eps ) )$ under a {\bf product distribution} as far as we are aware of (See \cite{watson_communication_2020} and references therein). And such a bound is what is necessary to satisfy \pref{def:hard}.

In this section, we develop novel technical tools to show 0-XOR function is ``hard" using the average min-entropy and its relation to KKL Theorem, which can be of independent interest. We suspect connections to lower bounds in other models of computation such as streaming (See Chapter 2 of \cite{roughgarden_communication_2016}).

\begin{theorem} \label{thm:min_main}
    Suppose $\varsigma_i [j]$'s are selected independently at random. There exists some $D = \Omega ( \log n )$ such that $f = \bigvee_{i=1}^{\ell} f_i$ with $f_i := \AND_{d = 1}^D \left( \ip{ \varsigma_i [d], T_i} \right)$ is hard as in \pref{def:hard}. 
\end{theorem}

Observe that balancedness of \pref{def:hard} is given for free due to the choice of $\varsigma_i$'s. All it remains to show is that for any $M$ of length $< n_i / \poly \log (n)$,
\begin{equation*}
    \E_{M} \left[ \odisc_{T_i |_{M=m} } ( \Psi_i ) \right] \leq n^{-2}.
\end{equation*}

Here is the rough outline of the proof. 
\begin{enumerate}
    \item By increasing the length of the message by attaching a message from $\cB$, ``spoiling" few bits, that is giving up on some $2^{-n^{1/3}}$ fraction of the input, and dividing the updates into blocks of size $K = \log^k (n)$, we can make the distribution easy for min-entropy (i.e. \pref{thm:strong_chain_rule_min}). This allows us to use (i) near chain-rule; (ii) KKL Theorem when conditioned on the message, two crucial components of the proof. We consider a set of ``good" messages $M$ and $\cB$, which again constitute all but exponential fraction of the input. (See \pref{lem:good_MB})
    \item For every choice of $d \in [D]$, using the chain-rule, we argue that since the min-entropy on parts of the input must be large, the min-entropy (over a random choice of blocks) must be large. We show a stronger statement which states that with all but $n^{-\Omega(1)}$ choice of blocks, the min-entropy must be large. (See \pref{lem:inductive_2})
    \item Using a variant of KKL Theorem in our context, we show that a large min-entropy over some small set of coordinates translates to small average bias. (See \pref{lem:KKL_applied})
    \item Iteratively apply step 2 and step 3 per $d \in [D]$. Via union bound over $d \in [D]$, this shows that for all but $n^{-\Omega(1)}$ choice of $\varsigma_i [d]$'s, the advantage must be $n^{-\Omega(1)}$ when conditioned on $M$ and $\cB$.
\end{enumerate}

\paragraph{Notations} 

We introduce the following notations for the proof. As we will divide the updates into blocks of size $K$ each, these blocks will be arranged so that $D$ columns and $R$ rows will have integer number of blocks. Recall that $\varsigma_i [d]$ corresponds to a vector which chooses exactly one entry per row. Denote $T_i^{(d, r, j) }$ as $j$-th block of $K := \log^k(n)$ updates associated with the column $d$, row $r$. There will be $D$ rows, $R$ columns. So there are $\frac{n_i}{D R K}$ many possible blocks for the index $j$. We denote $T_i^{(d, r) }$ as all the blocks in column $d$, row $r$. Then we denote $T_i^{(d) }$ as all the blocks in column $d$.

\subsubsection{Constructing $\cB$ for min-entropy} 

Without loss of generality, we can assume that our message $M$ is ``good," in a following sense. 
\begin{proposition} \label{prop:good_MB}
    Write $T_i := \left\{ T_i^{(d, r, j)} \right\}_{d \in [D], r \in [R], j \in [\frac{n_i}{DRK}]}$. Then there exists partitioning of $T_i |_{M = m}$, say $\cB$ of size 
    $\log |\cB| = O (  \frac{n_i \log n}{K} +  \frac{ k n_i \log \log n}{K} ) $ such that
    \begin{itemize}
        \item $\sum_{B \in \cB} \Pr [ B | M = m ] \geq 1 - \exp ( n^{-\Omega(1)} )$, that is the partition covers most of $T_i$.
        \item Conditioned on $B \in \cB$, block distributions $T^{j}_i | T^{< j}_i$ are nearly flat. That is
        $$ \forall T_i, T_i' \in B,~~ 2^{-O( n^{-1})} \leq \frac{\Pr [ T^{j}_i | T^{< j}_i , M = m , B ]}{\Pr [ T'^{j}_i | T'^{< j}_i , M = m , B ]} \leq 2^{O(n^{-1})} $$
        \item Suppose we give lexicographic ordering to $(d, r, j)$'s. For every $B \in \cB$, for every index $(d, r, j)$, and for every set $J$ whose elements are all less than $(d, r, j)$, we have
        \begin{enumerate}
            \item The chain-rule for min-entropy 
            \begin{align*} &\min_{t} \{ H_\infty ( T^{(d, r, j)}_i | T^J_i = t, B, M = m) \} + H_\infty ( T^J_i | B, M = m  ) \\
            & = H_\infty ( T^{(d, r, j)}_i, T^J_i | B, M = m ) \pm n^{-1} \end{align*}
            \item The average and worst-case min-entropy almost match
            $$ \widetilde{H}_\infty ( T^{(d, r, j)}_i | T^J_i, B, M = m ) = \min_{t} \{ H_\infty ( T^{(d, r, j)}_i | T^J_i = t, B, M = m ) \} \pm n^{-1} $$
        \end{enumerate}
    \end{itemize}
\end{proposition}
\begin{proof}
    This is a direct corollary of \pref{thm:strong_chain_rule_min} in our setting. By setting $\delta = n^{-1}, \eps = n^{-\Omega(1)}$, and $\log |\cX| = \log^k ( n )$ on the distribution of $T_i$ conditioned on $M= m$, \pref{thm:strong_chain_rule_min} implies the existence of such $\cB$ with the claimed parameters.
\end{proof}

Note that we can also make $\log |\cB| = n_i / \poly \log (n)$ via adjusting the block size parameter $K$ to match with the length of the message $M$. Thereby, we can attach the partitioning $B$ along with the message $M = m$ as our one-way message from Bob. Without loss of generality, we will denote $B_0$ as the ``spoiled" part, that is a set of $T_i$ not covered by $\cB$.

This implies the following lemma in our context, which states that the average min-entropy must be large when conditioned on $M=m$ and $B$.

\begin{lemma} \label{lem:good_MB}
    Recall that $|M| + |\cB| = O ( n_i \log n / K)$. 
    \begin{equation*}
        \Pr_{M,\cB} \left[ \exists d \in [D],~~ \widetilde{H}_\infty \left(  T_i^{(d) } | T^{(<d)}_i , M = m, B \right) \leq \frac{n_i}{D} - 2( |M| + |\cB| ) \right] \leq 2^{- n^{1/3} }
    \end{equation*}
\end{lemma}
\begin{proof}
    Due to \pref{lem:dors}, and the independence of $T_i^{(d) }$ and $T^{(<d)}_i$ (without conditioning on $M$ and $\cB$) we know that for any fixed $T^{(<d)}_i = t^{(<d)}_i$
    \begin{align*}
        \widetilde{H}_\infty \left(  T_i^{(d) } | T^{(<d)}_i = t^{(<d)}_i , M, \cB \right) & = -\log \E_{M,\cB, T^{(<d)}_i = t^{(<d)}_i } \left[ 2^{-H_\infty ( T_i^{(d) } | T^{(<d)}_i = t^{(<d)}_i, M= m , B )} \right] \\
        & \geq \frac{n_i}{D} - (|M| + |\cB|)
    \end{align*}
    which in turn implies
    \begin{align*}
         \E_{M,\cB, T^{(<d)}_i} \left[ 2^{-H_\infty ( T_i^{(d) } | T^{(<d)}_i = t^{(<d)}_i, M= m , B )} \right] \leq 2^{ - \frac{n_i}{D} + (|M| + |\cB|)}
    \end{align*}
    Due to a simple Markov's inequality, for a fixed $d \in [D]$ 
    \begin{equation} \label{eq:min_entropy_bound_fixed_D}
        \Pr_{M,\cB} \left[ \E_{T^{(<d)}_i | M= m ,\cB = B } \left[ 2^{-H_\infty ( T_i^{(d) } | T^{(<d)}_i = t^{(<d)}_i, M= m , B )} \right] \geq 2^{ - \frac{n_i}{D} + 2 (|M| + |\cB|)}  \right] \leq 2^{-(|M| + |\cB|)}
    \end{equation}
    Then applying union bound over $d \in [D]$ with \pref{eq:min_entropy_bound_fixed_D},
    \begin{equation*}
        \Pr_{M,\cB} \left[ \exists d \in [D],~~ \widetilde{H}_\infty \left(  T_i^{(d) } | T^{(<d)}_i , M = m, B \right) \leq \frac{n_i}{D} - 2( |M| + |\cB| ) \right] \leq D \cdot 2^{-(|M| + |\cB|)} \leq 2^{- n^{1/3} }
    \end{equation*} 
\end{proof}

We will denote the set of message and partition pair $M = m, B$ induced from \pref{lem:good_MB} as $\cG_{m,B}$. That is define
\begin{equation}\label{eq:good_MB}
    \cG_{m,B} := \left\{ (m, B) | \forall d \in [D],~~ \widetilde{H}_\infty \left(  T_i^{(d) } | T^{(<d)}_i , M = m, B \right) \geq \frac{n_i}{D} - 2( |M| + |\cB| ) \right\}.
\end{equation}
For the rest of the inductive arguments, we will only consider $(m,B)$ pairs from $\cG_{m,B}$.

\subsubsection{Large Min-Entropy per random choice of Blocks} 

    For brevity denote $O_{d}$ as the event $\AND_{j = 1}^{d-1} \left( \ip{ \varsigma_i [j], T_i} \right) = 1$. As inductive hypothesis, we assume that we have chosen $\varsigma_i[j]$'s for $j \in [d-1]$ such that
    \begin{equation} \label{eq:inductive_hypothesis}
        \Pr [ O_d | M = m, B  ] \in  \left[ \left(  \frac{1 - n^{-\Omega(1)} }{2} \right)^{d-1}, \left(  \frac{1 + n^{-\Omega(1)} }{2} \right)^{d-1} \right]
    \end{equation}
    Note that as a base case, when $d = 1$, the statement is trivially true, by assuming $\AND$ over null arguments to be 1.
    
\begin{lemma} \label{lem:inductive_2}
    Suppose $(m,B) \in \cG_{m,B}$ and as inductive hypothesis, assume \pref{eq:inductive_hypothesis} holds for $d \in [D]$.
    \begin{equation*}
        \Pr_{j_1, \ldots, j_R \in_{\cU} \left[ \frac{n_i}{ D R K } \right]  } \left[  \widetilde{H}_\infty \left( \left\{ T_i^{(d, r, j_r) } \right\}_{r = 1}^R |  T^{(<d)}_i , M = m, B, O_{d} \right) \leq ( 1 - \gamma ) KR  \right] \leq \exp ( - \Omega (\gamma^2 R ) )
    \end{equation*}
\end{lemma}
\begin{proof}

    We consider the following random variables (over the random choice of $j_r$'s) to apply McDiarmid's Inequality by considering the real-valued function 
    $$ h (j_1, \ldots, j_R) := \widetilde{H}_\infty \left( \left\{ T_i^{(d, r, j_r) } \right\}_{r = 1}^R |  T^{(<d)}_i , M = m, B, O_{d} \right). $$

    For completeness, we state the McDiarmid's Inequality here.
    \begin{fact}[McDiarmid's Inequality] \label{fact:mcdiarmid}
        Let $f : \cX_1 \times \ldots \times \cX_n \to \mathbb{R}$ satisfy the following property: there exists $c_1, \ldots, c_n$ such that for all $x_1 \in \cX_1, \ldots, x_n \in \cX_n$
        \begin{equation*}
            \sup_{x'_i \in \cX_i} \abs{f(x_1, \ldots, x_i, \ldots, x_n ) - f(x_1, \ldots, x'_i, \ldots, x_n ) } \leq c_i 
        \end{equation*}
        Then if $X_1, \ldots, X_n$ are chosen independently at random from $\cX_1, \ldots, \cX_n$ respectively, 
        \begin{equation*}
            \Pr \left[ f(X_1, \ldots, X_n) - \E_{X_1, \ldots, X_n} [f(X_1, \ldots, X_n) ] \leq - \eps \right ] \leq \exp \left( \frac{2 \eps^2}{\sum_{i=1}^n c_i^2} \right)
        \end{equation*}
    \end{fact}

    Note that due to the chain rule property, replacing a single coordinate can only lead to a difference of at most $K \pm n^{-1}$ in $h$. Therefore \pref{fact:mcdiarmid} directly implies
    \begin{align}
        \label{eq:mcdiarmid_direct}
         \Pr_{j_1, \ldots, j_R \in_{\cU}} \left[ h (j_1, \ldots, j_r) \leq \E [ h ]  - \alpha \right] \leq \exp \left( - \frac{2 \alpha^2}{R (K+n^{-1})^2} \right).
    \end{align}
    To transform \pref{eq:mcdiarmid_direct} to a desired form, we calculate the expected value of $h$. Observe that we can decompose $h$ (upto $n^{-1}$ error) as
    \begin{equation}
        h (j_1, \ldots, j_R) = \sum_{r=1}^R \left( \underbrace{ \widetilde{H}_\infty \left( T_i^{(d, r, j_r) } | T_i^{(d, <r, j_{<r}) }, T^{(<d)}_i , M = m, B, O_{d} \right) }_{ := h_r (j_1, \ldots, j_R) } \pm \frac{1}{n} \right)
    \end{equation}

    Observe that $h_r$ only depends on $j_1, \ldots, j_r$. Consider a fixed $r$. For any choice of $j_1, \ldots, j_{r-1}$, 
    \begin{align*}
        \E_{j_r} [h_r] & = \frac{DRK}{n_i} \sum_{j_r} \widetilde{H}_\infty \left(  T_i^{(d, r, j_r) } |  T_i^{(d, <r, j_{<r}) }, T^{(<d)}_i , M = m, B, O_{d} \right) \\
        & \geq \frac{DRK}{n_i} \widetilde{H}_\infty \left(  T_i^{(d, r) } |  T_i^{(d, <r, j_{<r}) }, T^{(<d)}_i , M = m, B, O_{d} \right) \pm \frac{1}{n}\\
        & \geq \frac{DRK}{n_i} \widetilde{H}_\infty \left(  T_i^{(d, r) } |  T_i^{(d, <r ) }, T^{(<d)}_i , M = m, B, O_{d} \right) \pm \frac{1}{n}
    \end{align*}
    where the first inequality follows from \pref{cl:conditioning_infty} and the chain-rule property induced by $\cB$. The second inequality follows from \pref{cl:conditioning_infty}.

    \begin{align}
        \E [ h ] & = \sum_{r=1}^R \E [ h_r ]  \geq  \frac{DRK}{n_i} \sum_{r=1}^R  \left( \widetilde{H}_\infty \left(  T_i^{(d, r) } |  T_i^{(d, <r ) }, T^{(<d)}_i , M = m, B, O_{d} \right) - 2 n^{-1} \right) \nonumber \\
        & \geq \frac{DRK}{n_i} \left( \widetilde{H}_\infty \left(  T_i^{(d) } | T^{(<d)}_i , M = m, B, O_{d} \right) - 4 R n^{-1} \right) \nonumber \\
        & \geq \frac{DRK}{n_i} \left(  \widetilde{H}_\infty \left(  T_i^{(d) } | T^{(<d)}_i , M = m, B \right) - 4 R n^{-1} - \log \frac{1}{\Pr [ O_d | M = m, B  ]} \right) 
    \end{align}
    where the last bound holds due to the definition of $O_d$ and average min-entropy, which implies
    \begin{align*}
        \Pr [ O_d | M = m, B  ] \cdot 2^{- \widetilde{H}_\infty \left( T_i^{(d) }  | T^{(<d)}_i ,  M = m, B, O_{d} \right)} \leq  2^{- \widetilde{H}_\infty \left( T_i^{(d) } | T^{(<d)}_i,  M = m, B \right)}.
    \end{align*}

    Recall that due to our assumption that $(m,B) \in \cG_{m,B}$,
    \begin{equation*}
        \widetilde{H}_\infty \left( T_i^{(d) }  | T^{(<d)}_i ,  M = m, B \right) = \frac{n_i}{D} - 2 ( | M | + |\cB| ) \geq \frac{n_i}{D} - O ( n_i \log n / K )
    \end{equation*}
    as well, then 

    \begin{equation*}
        \E [ h ] \geq RK \underbrace{ -  O( D R \log n ) - \frac{ 4 D R^2 K }{n \cdot n_i}  - \frac{ D R K }{n_i} \cdot \log \frac{1}{\Pr [ O_d | M = m, B  ]} }_{ - o ( R K )}.
    \end{equation*}
    This implies that
    \begin{align*}
         & \Pr_{j_1, \ldots, j_R \in_{\cU}} \left[ h (j_1, \ldots, j_r) < RK  - \gamma R K \right]  \\
         & \leq \Pr_{j_1, \ldots, j_R \in_{\cU}} \left[ h (j_1, \ldots, j_r) < \E [ h ]  - ( \gamma + o(1) ) R K \right]  
         \leq \exp \left( - \frac{3}{2} \gamma^2 R \right)
    \end{align*}
    completing the proof of the lemma.
\end{proof}

\subsubsection{Large Min-Entropy Implies Small Average Bias}  

We would like to then show that assuming $M = m, B, O_{d}$, the ``average" bias over chosen coordinates are small using \pref{lem:kkl}. 

Suppose we take the set $\varsigma_i [d]$ as those induced by selecting one coordinate per row $r$, but restricted to be in $j_r$ block of the input. Then there are $K^R$ possible choices of such $\varsigma_i [d]$. Consider a uniform distribution over such $\varsigma_i [d]$ as $\cU_{j_1, \ldots, j_R}$. A simple observation here is that the original uniform distribution $\cU$ over $\varsigma_i [d]$ (i.e. uniformly selecting one entry per row at random) can be decomposed as
\begin{equation*}
    \E_{j_1, \ldots, j_R \in_{\cU} \left[ \frac{n_i}{ D R K } \right]  } \left[ \cU_{j_1, \ldots, j_R} \right] = \cU
\end{equation*}
Now suppose we consider $(j_1, \ldots, j_R)$ that satisfies \pref{lem:inductive_2}. Denote such $(j_1, \ldots, j_R)$ as $\cJ_{d}$. That is 
\begin{equation}
    \cJ_{d} := \left\{ (j_1, \ldots, j_R) | \widetilde{H}_\infty \left( \left\{ T_i^{(d, r, j_r) } \right\}_{r = 1}^R |  T^{(<d)}_i , M = m, B, O_{d} \right) \geq ( 1 - \gamma ) KR  \right\}
\end{equation}

We can show the following lemma for $(j_1, \ldots, j_R) \in \cJ_{d}$ using \pref{lem:kkl}.

\begin{lemma} \label{lem:KKL_applied}
    Assume the premise of \pref{lem:inductive_2}. Suppose $(j_1, \ldots, j_R) \in \cJ_{d}$. Then
    \begin{equation*}
        \E_{\varsigma_i [d] \sim \cU_{j_1, \ldots, j_R} } \left[ \E \left[ \chi_S ( \left\{ T_i^{(d, r, j_r) } \right\}_{r = 1}^R  ) | M=m, B, O_d \right]^2 \right]  \leq 2^{ R \log ( 4 \gamma) }
    \end{equation*}
\end{lemma}
\begin{proof}
    By the property guaranteed by \pref{thm:strong_chain_rule_min}, \pref{prop:good_MB}, we know that  %% near uniform distribution
    \begin{align}
        \forall t, t' \in \supp ( \left\{ T_i^{(d, r, j_r) } \right\}_{r = 1}^R |_{M=m, B, O_d} ),~~ \frac{\Pr\left[ \left\{ T_i^{(d, r, j_r) } \right\}_{r = 1}^R = t | M=m, B, O_d \right] }{\Pr \left[ \left\{ T_i^{(d, r, j_r) } \right\}_{r = 1}^R = t' | M=m, B, O_d \right] } \in 1 \pm O ( R n^{-1} ) \label{eq:uniform_guarantee}
    \end{align}

    Thus, if we simply consider the function $g : \left\{ T_i^{(d, r, j_r) } \right\}_{r = 1}^R \to \mathbb{R}^{\geq 0}$ as
    \begin{equation*}
        g ( x ) := \Pr \left[ \left\{ T_i^{(d, r, j_r) } \right\}_{r = 1}^R = x | M=m, B, O_d \right] \cdot \abs{ \supp ( \left\{ T_i^{(d, r, j_r) } \right\}_{r = 1}^R |_{M=m, B, O_d} ) },
    \end{equation*} 
    as \pref{eq:uniform_guarantee} implies that
    \begin{equation*}
        \Pr \left[ \left\{ T_i^{(d, r, j_r) } \right\}_{r = 1}^R = x | M=m, B, O_d \right] \leq \frac{1 + O( R n^{-1} )}{\abs{ \supp ( \left\{ T_i^{(d, r, j_r) } \right\}_{r = 1}^R |_{M=m, B, O_d} ) }}
    \end{equation*}
    we have the guarantee that for any $x$ in the support of $g$, 
    \begin{equation*}
         g (x) \leq 1 + O ( R n^{-1} ).
    \end{equation*}

    Then \pref{lem:kkl} implies that for any $\delta > 0$,
    \begin{align}
        & \sum_{S \subset [R K]} \delta^{|S|} \cdot \widetilde{g} (S)^2 \leq ( 1 + O(Rn^{-1}) )^2 \left( \frac{2^{RK}}{\abs{ \supp ( \left\{ T_i^{(d, r, j_r) } \right\}_{r = 1}^R |_{M=m, B, O_d} )} } \right)^{2\delta}. \label{eq:KKL_step1}
    \end{align}
    Observe that $\widetilde{g} (S)$ term is exactly 
    \begin{align*}
        \widetilde{g} (S) & = \frac{1}{2^{RK}}\sum_{x \in \{0,1\}^{RK}} g ( x ) \cdot \chi_S (x) \frac{2^{RK}}{\abs{ \supp ( \left\{ T_i^{(d, r, j_r) } \right\}_{r = 1}^R |_{M=m, B, O_d} )} } \\
        & = \E \left[ \chi_{S} (\left\{ T_i^{(d, r, j_r) } \right\}_{r = 1}^R) | M=m, B, O_d \right],
    \end{align*}
    while we can bound 
    \begin{align*}
        \frac{2^{RK}}{\abs{ \supp ( \left\{ T_i^{(d, r, j_r) } \right\}_{r = 1}^R |_{M=m, B, O_d} )} } \leq 2^{ ( R K - H_\infty ( \left\{ T_i^{(d, r, j_r) } \right\}_{r = 1}^R | M=m, B, O_d ) ) }
    \end{align*}
    due to \pref{cl:min_entropy_supp}. As $(j_1, \ldots, j_R) \in \cJ_{d}$, and due to \pref{cl:conditioning_infty}, we can rewrite \pref{eq:KKL_step1} as

    \begin{align*}
        & \sum_{S \subset [R K]} \delta^{|S|} \E \left[ \chi_S ( \left\{ T_i^{(d, r, j_r) } \right\}_{r = 1}^R  ) | M=m, B, O_d \right]^2 \\
        & \leq ( 1 + O(R n^{-1}) )^2 2^{ 2 \delta \cdot \gamma R K }
    \end{align*}
    Select $\delta$ as $\frac{1}{2 \gamma K}$. Then by considering $\varsigma_i[d] \sim \cU_{j_1, \ldots, j_R}$, 
    \begin{equation*}
        \E_{\varsigma_i [d] \sim \cU_{j_1, \ldots, j_R} } \left[ \E \left[ \chi_{\varsigma_i [d]} ( \left\{ T_i^{(d, r, j_r) } \right\}_{r = 1}^R  ) | M=m, B, O_d \right]^2 \right] \leq 2^{ 1 + R + R \log ( 2 \gamma K ) - R \log K} = 2^{ 1 + R \log ( 4 \gamma) }.
    \end{equation*}

\end{proof}

A simple Markov's inequality implies the following corollary, which we will use towards our main proof.

\begin{corollary} \label{cor:kkl_applied}
    Assume the premise of \pref{lem:inductive_2} and $(j_1, \ldots, j_R) \in \cJ_{d}$. Then
    \begin{equation*}
    \Pr_{\varsigma_i [d] \sim \cU_{j_1, \ldots, j_R} } \left[ \abs{ 2 \Pr \left[  O_{d+1}  | M=m, B, O_d \right] - 1 } \geq  2^{ \frac{R \log (8 \gamma)}{2} } \right] \leq 2^{1 - 2R}
    \end{equation*} 
\end{corollary}
\begin{proof}
    Observe that 
    \begin{equation*}
        \abs{ 2 \Pr \left[  O_{d+1}  | M=m, B, O_d \right] - 1 } = \abs{ \E \left[ \chi_{\varsigma_i [d]} ( \left\{ T_i^{(d, r, j_r) } \right\}_{r = 1}^R  ) | M=m, B, O_d \right] }
    \end{equation*}
    and that \pref{lem:KKL_applied} along with Markov's inequality implies
    \begin{equation*}
    \Pr_{\varsigma_i [d] \sim \cU_{j_1, \ldots, j_R} } \left[ \E \left[ \chi_{\varsigma_i [d]} ( \left\{ T_i^{(d, r, j_r) } \right\}_{r = 1}^R  ) | M=m, B, O_d \right]^2 \geq  2^{R \log (8 \gamma)} \right] \leq  \frac{2^{1 + R \log (4 \gamma)}}{2^{R \log (8 \gamma)}} = 2^{1- 2R}
    \end{equation*}  
\end{proof}

    Taking $\gamma$ to be a sufficiently small constant, while taking large enough $R = \Theta ( \log n )$ , we obtain the bounds for \pref{lem:inductive_2} and \pref{cor:kkl_applied} with
    \begin{align*}
        & \exp ( - 1.5 \gamma^2 R ) \leq n^{-100}, ~~ 2^{\frac{R \log (8 \gamma)}{2}} \leq o(n^{-2}), ~~ 2^{1-2R} \leq n^{-100}.
    \end{align*}

\subsubsection{Combining All}

Applying \pref{lem:inductive_2} and \pref{lem:KKL_applied} iteratively leads to the following lemma.
\begin{lemma} \label{lem:min_main}
\begin{equation*}
    \Pr_{\varsigma_i [1], \ldots, \varsigma_i [D] \sim \cU^D} \left[ \Pr_{T_i | M=m, B} [ O_D | M = m, B  ]   \notin \left[ \left(  \frac{1 - n^{-2} }{2} \right)^{D}, \left(  \frac{1 + n^{-2} }{2} \right)^{D} \right] \right] \leq n^{-99}
\end{equation*}
\end{lemma}
\begin{proof}
    Let 
    \begin{equation*}
        \cF_{d} := \left\{ \varsigma_i [d] : \Pr_{T_i | M=m, B} [ O_{d+1} | M = m, B, O_d ]  \in \left[ \left(  \frac{1 - n^{-2} }{2} \right), \left(  \frac{1 + n^{-2} }{2} \right) \right] \right\} 
    \end{equation*} 
    If we choose $\varsigma_i [1] \in \cF_{1} , \ldots, \varsigma_i [d-1] \in \cF_{d-1}$, 
    \begin{equation*}
        \Pr_{T_i | M=m, B} [ O_d | M = m, B  ] \in \left[ \left(  \frac{1 - n^{-2} }{2} \right)^{d-1}, \left(  \frac{1 + n^{-2} }{2} \right)^{d-1} \right].
    \end{equation*}
    Then for such choice of $\varsigma_i [1], \ldots, \varsigma_i [d-1]$, \pref{lem:inductive_2} and \pref{cor:kkl_applied} implies that 
    \begin{equation*}
    1 - \Pr_{ \varsigma_i [d]} \left[ \cF_{d} \right] \leq 2 n^{-100}
    \end{equation*}
    and furthermore, choosing $\varsigma_i [d]$ from $\cF_{d}$ would further give
    \begin{equation*}
        \Pr_{T_i | M=m, B} [ O_{d+1} | M = m, B  ] \in \left[ \left(  \frac{1 - n^{-2} }{2} \right)^{d}, \left(  \frac{1 + n^{-2} }{2} \right)^{d} \right].
    \end{equation*}
    Iteratively applying the argument, if $\forall d \in [D]$, $\varsigma_i [d] \in \cF_d$,
    \begin{equation*}
        \Pr_{T_i | M=m, B} [ O_D | M = m, B  ]   \in \left[ \left(  \frac{1 - n^{-2} }{2} \right)^{D}, \left(  \frac{1 + n^{-2} }{2} \right)^{D} \right].
    \end{equation*}
    As each $\varsigma_i [d]$'s are chosen independently at random, the probability of
    \begin{equation*}
        \Pr_{\varsigma_i [1], \ldots, \varsigma_i [D]}  \left[ \forall d \in [D], \varsigma_i [d] \in \cF_d \right] \geq 1 - 2 D n^{-100} \geq 1 - n^{-99}
    \end{equation*}
    as $D$ is some linear factor of $\log n$, which completes the proof of the lemma.
\end{proof}

Now we are ready to complete the proof of \pref{thm:min_main} with \pref{lem:min_main}.

\begin{proofof}{\pref{thm:min_main}}

We would like to show that $\Psi_i$ induced by $f_i$ and underlying $f$ satisfies \pref{def:hard}. For any setting of $\varsigma_i$'s over $i \in [\ell]$,
\begin{equation*}
    \Pr_{T} \left[ \OR ( f_1 (T_1) , \ldots f_{i-1} ( T_{i-1}), f_{i+1} ( T_{i+1}), f_\ell ( T_{\ell} ) ) = +1 \right] = 1 - \left( 1 - 2^{-D} \right)^\ell.
\end{equation*}
The choice of $D$ was ensured to guarantee that the above quantity is $\Omega(1)$, that is with probability $g_i \geq \Omega(1)$, $f_i ( T_i )$ matters. Furthermore, for any setting of $\varsigma_i$,
\begin{equation*}
    \Pr_{T_i} \left[ f_i (T_i) = +1 \right] = 2^{-D} = \Theta ( \ell^{-1} ),
\end{equation*}
making the function balanced with the balanced parameter $\beta = \Theta ( \ell^{-1} )$. This is good enough as $\beta \geq 2^{- o ( \sqrt{ \log n })}$. Therefore, our $f$ satisfies the first condition of \pref{def:hard}.

Now observe that due to normalization
\begin{equation*}
    \Psi_i ( \varsigma_i , T_i ) = \begin{cases}
        + 1& \mbox{if } f_i ( \varsigma_i, T_i ) = +1 \\ 
        - \frac{1}{1-2^{-D}} + 1 = - \frac{2^{-D}}{1 - 2^{-D}} & \mbox{otherwise}  
    \end{cases}
\end{equation*}

We would like to show that the above matrix $\Psi_i$ has small $\odisc$ for $|M| = c \leq n / \poly \log (n)$ where
\begin{align*}
    & \odisc_{T_i |_{M=m}} ( \Psi_i ) := \E_{\varsigma_i} \abs{ \sum_{T_i}  \Psi_i ( \varsigma_i, T_i ) \cdot \Pr[ T_i | M = m ] } \\
    & \odisc_{c} ( \Psi_i ) := \max_{|M| \leq c} \E_{M} \left[ \odisc_{T_i |_{M=m}} ( \Psi_i ) \right]
\end{align*} 

Recall that for any $M$, we can create a further partitioning $\cB$, which would then
\begin{align*}
    & \E_{B \in \cB} \left[ \odisc_{T_i |_{M=m,B} } ( \Psi_i ) \right] = \E_{B \in \cB} \left[ \E_{\varsigma_i} \abs{ \sum_{T_i}  \Psi_i ( \varsigma_i, T_i ) \cdot \Pr[ T_i | M = m , B ] } \right] \\
    & \geq \E_{\varsigma_i} \abs{ \E_{B \in \cB} \left[ \sum_{T_i}  \Psi_i ( \varsigma_i, T_i ) \cdot \Pr[ T_i | M = m , B ] \right] } \geq \odisc_{T_i |_{M=m}} ( \Psi_i ) - \exp ( n^{-\Omega(1)} )
\end{align*}
where the last bound follows from 
\begin{equation*}
    \sum_{T_i} \abs{ \E_{B \in \cB} \left[ \Pr[ T_i | M = m , B ] \right] - \Pr[T_i | M = m ] } \leq \exp ( n^{-\Omega(1)} )
\end{equation*}
due to the property of $\cB$. 

Our goal is then to bound $\E_{M} \left[ \E_{B \in \cB} \left[ \odisc_{T_i |_{M=m,B} } ( \Psi_i ) \right] \right] $ as 
\begin{align*}
    \E_{M} \left[  \odisc_{T_i |_{M=m} } ( \Psi_i ) \right] \leq \E_{M} \left[ \E_{B \in \cB} \left[ \odisc_{T_i |_{M=m,B} } ( \Psi_i ) \right] \right] + \exp ( n^{-\Omega(1)} ) 
\end{align*}
We consider $\odisc_{T_i |_{M=m,B} } ( \Psi_i )$ for some fixed $M=m,B$. Suppose $M=m,B$ satisfies
\begin{equation} \label{eq:good_MB_condition_thmproof}
    \forall d \in [D],~~ \widetilde{H}_\infty \left(  T_i^{(d) } | T^{(<d)}_i , M = m, B \right) \leq \frac{n_i}{D} - 2( |M| + |\cB| ).
\end{equation}
\pref{lem:min_main} implies that for such choice of $M=m,B$,
\begin{equation} \label{eq:good_MB_thmproof}
    \Pr_{\varsigma_i [1], \ldots, \varsigma_i [D]} \left[ \Pr_{T_i | M=m, B} [ O_D | M = m, B  ]   \notin \left[ \left(  \frac{1 - n^{-2} }{2} \right)^{D}, \left(  \frac{1 + n^{-2} }{2} \right)^{D} \right] \right] \leq n^{-99}
\end{equation}
As per choice of $\varsigma_i [1], \ldots, \varsigma_i [D]$, the advantage can be written as
\begin{align}
    & \abs{ \Pr_{T_i | M=m, B} [ O_D | M = m, B  ] - \frac{2^{-D}}{1 - 2^{-D}} ( 1 - \Pr_{T_i | M=m, B} [ O_D | M = m, B  ] ) } \nonumber \\
    & = \frac{1}{1 - 2^{-D}} \abs{ \Pr_{T_i | M=m, B} [ O_D | M = m, B  ] - 2^{-D} } \label{eq:advantage_min}
\end{align}
\pref{eq:good_MB_thmproof} implies that with probability all but $n^{-99}$ over $\varsigma_i [1], \ldots, \varsigma_i [D]$, \pref{eq:advantage_min} is then at most 
\begin{equation*}
    \pref{eq:advantage_min} \leq \frac{2^{-D}}{1 - 2^{-D}} \left(  \left(  1 + n^{-2} \right)^{D} - 1 \right) \leq o( n^{-2} )
\end{equation*}
which then concludes that if $M=m,B$ satisfies \pref{eq:good_MB_condition_thmproof},
\begin{equation*}
    \odisc_{T_i |_{M=m,B} } ( \Psi_i ) \leq o ( n^{-2} )
\end{equation*}
\pref{lem:good_MB} then implies that the fraction of $M=m,B$ that does not satisfy \pref{eq:good_MB_condition_thmproof} is at most $2^{-n^{1/3}}$.
Therefore, 
\begin{align*}
    \E_{M} \left[  \odisc_{T_i |_{M=m} } ( \Psi_i ) \right] \leq \E_{M} \left[ \E_{B \in \cB} \left[ \odisc_{T_i |_{M=m,B} } ( \Psi_i ) \right] \right] + \exp ( n^{-\Omega(1)} ) \leq o(n^{-2}) 
\end{align*}
\end{proofof}

\newpage 
\bibliographystyle{alpha}
\bibliography{references.bib}

\end{document}